\documentclass{article}
\usepackage{arxiv}

\usepackage{amsthm}
\usepackage{amsmath}
\usepackage{dsfont}
\usepackage{amssymb}
\usepackage{natbib}
\usepackage{float}
\usepackage[colorlinks,citecolor=blue,urlcolor=blue,filecolor=blue,backref=page]{hyperref}
\usepackage{graphicx}
\usepackage{enumitem}

\theoremstyle{plain}
\newtheorem{theorem}{Theorem}[section]
\newtheorem{lemma}[theorem]{Lemma}
\newtheorem{condition}[theorem]{Condition}
\theoremstyle{definition}

\newcommand{\indep}{\perp\!\!\!\!\perp}
\title{Bayesian optimal experimental design for inferring causal structure}

\author{
  Michele Zemplenyi \\
  Department of Biostatistics\\
  Harvard University\\
  \texttt{mzemplenyi@g.harvard.edu} \\
   \And
 Jeffrey W. Miller \\
  Department of Biostatistics\\
  Harvard University\\
  \texttt{jwmiller@hsph.harvard.edu} \\
}

\begin{document}
\maketitle

\begin{abstract}
Inferring the causal structure of a system typically requires interventional data, rather than just observational data.
Since interventional experiments can be costly, it is preferable to select interventions that yield the maximum amount of information about a system. We propose a novel Bayesian method for optimal experimental design by sequentially selecting interventions that minimize the expected posterior entropy as rapidly as possible. A key feature is that the method can be implemented by computing simple summaries of the current posterior, avoiding the computationally burdensome task of repeatedly performing posterior inference on hypothetical future datasets drawn from the posterior predictive. After deriving the method in a general setting, we apply it to the problem of inferring causal networks. We present a series of simulation studies in which we find that the proposed method performs favorably compared to existing alternative methods.  Finally, we apply the method to real and simulated data from a protein-signaling network.
\end{abstract}

\keywords{optimal experimental design \and active learning \and graphical models}
%% *** Frontmatter *** 

% \begin{frontmatter}
% \title{Bayesian optimal experimental design for inferring causal structure}

% %\title{\thanksref{T1}}
% %\thankstext{T1}{<thanks text>}
% \runtitle{Bayesian optimal experimental design}

% \begin{aug}
% %\author{\fnms{} \snm{}}
% \author{\fnms{Michele} \snm{Zemplenyi}\thanksref{addr1}\ead[label=e1]{mzemplenyi@g.harvard.edu}}
% \and
% \author{\fnms{Jeffrey W.} \snm{Miller}\thanksref{addr1}\ead[label=e2]{jwmiller@hsph.harvard.edu}}

% \runauthor{Zemplenyi and Miller}

% \address[addr1]{Harvard School of Public Health, 677 Huntington Ave, Boston, MA 02115
%     \printead{e1} % print email address of "e1"
%     \printead*{e2}
% }

% % \thankstext{}{<text>}

% \end{aug}

%% ** Mainmatter **

%\section{}\label{}
\section{Introduction}\label{introduction}
Inferring the causal structure of a set of related variables is key to understanding how a system works. 
By interpreting directed edges as implying causal relationships, a causal network model extends standard (non-causal) graphical models by specifying the distribution of the data when one or more variables are manipulated \citep{Pearl}. A large body of work exists on learning the structure of graphical models from observational data, that is, data passively collected without any experimental intervention performed on the system under study; see \cite{Daly2011} for a comprehensive review.  However, typically, observational data alone can only reveal the structure of a graphical model up to the Markov equivalence class containing the true data-generating graph \citep{Verma91}. To fully determine the structure of a causal network without making additional assumptions about specific functional model classes and error distributions, interventional data are needed to resolve the directionality of uncompelled edges \citep{Peters2011}. Interventional data result from experiments in which one or more nodes have been actively manipulated, for example, by activating or inhibiting the expression of a gene in a model organism \citep{Sachs05,Nagarajan2013}.

Crucially, different intervention experiments yield different amounts of information about the causal structure. Thus, since experiments are often expensive and time-consuming, it is advantageous to select interventions that provide the maximum amount of information. Optimal experimental design (OED) methods, also referred to as active learning algorithms, attempt to optimize this experiment selection process by providing a means of evaluating which experiments should be performed next given the current state of knowledge. From the Bayesian perspective, a naive approach would be as follows: for each candidate experiment, generate hypothetical datasets from the posterior predictive, perform posterior inference on each dataset, and compute a functional of the posterior that summarizes the amount of information gained. 
Averaging over many datasets would yield an estimate of the posterior expected amount of information gain for each candidate experiment.  
However, this naive approach would involve an inordinate amount of computation.

In this article, we develop a novel Bayesian OED technique that is principled and computationally tractable.
Roughly speaking, we consider the asymptotic information gain that each experiment would yield in the limit of infinitely many replicates,
as a proxy for the expected gain from finitely many replicates.
Under fairly general conditions, in this limit, the posterior is simply obtained by restricting the current posterior to a subset of the parameter space.
Thus, it turns out the reduction in entropy can be easily computed using samples from the current posterior, without generating or performing inference on any hypothetical datasets.  This leads to a vast reduction in the computation burden required to select experiments.

Based on this principle, we introduce a class of entropy-based criteria for determining the optimal intervention to perform in the next experiment. 
%Roughly speaking, we propose to select interventions that would lead to the largest expected decrease in the posterior entropy of a feature of interest.  
After the selected experiment is performed and new experimental data is obtained, we update the posterior on graphs and use it to select the next experiment. To sample from the posterior distribution over graphs, we employ an existing Markov chain Monte Carlo (MCMC) algorithm with efficient dynamic programming-based proposals \citep{Eaton07_hybridMCMC}. By iterating between experimentation and analysis in this cyclical fashion, we focus the data collection efforts in a way that reduces posterior uncertainty as rapidly as possible. We compare our method to two other active learning approaches and a random intervention approach in the context of several simulated data sets as well as the Sachs cell-signaling network, a commonly studied benchmark in the causal network literature \citep{Sachs05, Eaton07_DP, Cho, Ness}.

% Although the infinite-replicates gain overestimates the finite-replicates gain, we only need to optimal experiment to be the same under finite and infinite replicates in order for our criterion to select the experiment that is optimal for the finite-replicates setting.  Empirically, we find that our criterion often selects highly informative experiments.

The article is organized as follows. In Section \ref{sec:general-criterion}, we derive our general criterion for selecting optimal experiments. In Section \ref{sec:criterion-causal-network-models}, we apply our general criterion to causal network models. In Section \ref{sec:practical-implementation}, we lay out the overall proposed framework, along with implementation details about the entropy-based criteria and the MCMC algorithm. Section \ref{sec:previous-work} discusses related previous work on OED and active learning methods. In Section \ref{sec:simulation-results}, we present a collection of simulation studies. Section \ref{sec:application} contains an application to the Sachs network, using both real experimental data and simulated data. We conclude with a brief discussion of our findings and directions for further research.

\section{General criterion}\label{sec:general-criterion}
\subsection{Intuition}\label{subsec:intuition}
Before discussing OED in the context of causal networks, we first consider the more general case of identifying an object of interest from a large set of possible objects by asking a sequence of questions. For intuition, we illustrate the basic idea of our method in terms of the popular game Twenty Questions. In this game one person, the ``answerer," thinks of an object. The other player, the ``questioner," then asks a sequence of ``yes" or ``no" questions with the goal of guessing the answerer's object using fewer than twenty questions. %Upon learning the answer to each question, the questioners gain information that allows them to eliminate a number of possible objects. Additionally, each answer informs the next question that is asked. 

At the beginning of the game, the questioner has a prior over objects, representing the probability that the answerer has selected a given object. A question such as, ``Is the object living?" partitions the objects into two parts: living and non-living. The subsequent answer provides information that allows the questioner to eliminate the objects in one of the parts and update their posterior beliefs accordingly. If the prior is uniform, then the most efficient strategy is to select questions that partition the set of remaining objects roughly in half (Figure \ref{fig:partitionGeneral}). More generally, if the prior is not uniform, then it is most efficient to split the posterior probability roughly in half at each step.

\begin{figure}
    \centering
    \includegraphics[width=0.6\textwidth]{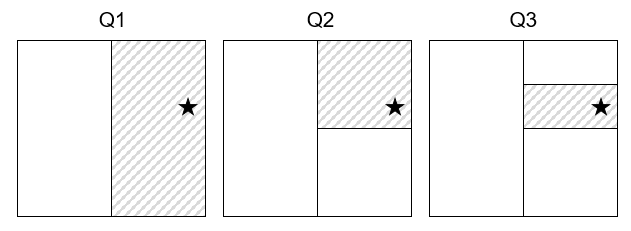}
    \caption{Schematic illustration of sequential partitioning by a series of questions, Q1-Q3. The answer to each question indicates the part (shaded region) containing the object of interest (denoted by a star.)}
    \label{fig:partitionGeneral}
\end{figure}

The following observations about Twenty Questions will be helpful to consider when we discuss the causal network setting in Section \ref{sec:criterion-causal-network-models}.  First, there are many possible ways of partitioning a space into two parts of equal size (or equal posterior probability, more generally). Different questions partition the space along different features of the objects. Second, we might relax the restriction to yes/no questions and also allow questions such as, ``Is the object a vegetable, animal, or mineral?" Such questions partition the space of objects into more than two parts. In choosing what question to ask next, the questioner is implicitly considering the informativeness of the partition induced by their question.

\subsection{General criterion}
\label{subsec:general-criterion}

We now formalize the ideas laid out in Section \ref{subsec:intuition}. Suppose $(\theta,\nu)\sim\pi$, $X_1,\ldots,X_N|\theta,\nu\sim P_{\theta,\nu}$ i.i.d., and $f(\theta)$ is a function of $\theta$ with the following three properties.
\begin{condition}\label{condition:general} ~
\begin{enumerate}[label=(\alph*), ref=\ref{condition:general}(\alph*)]
    \item\label{condition:cond-indep} $\theta \indep X_{1:N} \mid f(\theta)$.
    \item\label{condition:identifiable} $f(\theta)$ is identifiable, in the sense that there is a function $g$ such that $g(P_{\theta,\nu}) = f(\theta)$ almost surely, and
    \item\label{condition:finite} $f(\theta)$ can only take one of finitely many values.
\end{enumerate} 
\end{condition}
For interpretation, in the Twenty Questions illustration, $\theta$ represents the object selected by the answerer, $\pi$ is the prior distribution, $f(\theta)$ is the true answer to question $f$, and $X_1,\ldots,X_N$ represent noisy answers to question $f$. Note that $X_1,\ldots,X_N$ all pertain to the same question about $\theta$, not $N$ different questions.
Condition \ref{condition:cond-indep} means that once we know the true answer to the question, the noisy answers provide no additional information about $\theta$.
Condition \ref{condition:identifiable} means that the answer $f(\theta)$ is uniquely determined by the distribution of $X_n$, and thus, in the limit as $N\to\infty$, $f(\theta)$ can be recovered from $X_1,\ldots,X_N$.

In an experimentation context, $\theta$ is a parameter of interest, $\nu$ is a nuisance parameter, $\pi$ is the prior, $f(\theta)$ is the answer to a research question $f$ (for example, is a certain hypothesis true),
and $X_1,\ldots,X_N$ are data from $N$ replicates of an experiment performed to obtain information about $f(\theta)$.  
%Unlike in Twenty Questions, where the answerer will generally respond with certainty about $f(\theta)$, each replicate of the experiment yields noisy data that is informative about $f(\theta)$ but not definitive.
The nuisance parameter $\nu$ is sometimes needed for the i.i.d.\ assumption to hold.
For the causal network models that we consider in subsequent sections, $\theta$ is a directed acyclic graph and $f$ corresponds to an equivalence class of graphs.
% we will show in Section \ref{sec:markov-tse-equivalent} that Condition \ref{condition:general} is satisfied for a particular $f$ defined by interventional Markov equivalence classes \citep{Hauser2012, Tian2001}.
First, however, we consider a general model and any $f$ that satisfies Condition \ref{condition:general}.

\subsubsection{Approximate information gain}\label{subsec:approx-info-gain}

The entropy of a random variable $Y$ is defined as $H(Y) := -\int p(y)\log p(y) d\mu(y)$ where $p(y)$ is the density of $Y$ with respect to some dominating measure $\mu$, or more succinctly, $H(Y) = - \mathrm{E}(\log p(Y))$.
Similarly, $H(Y|Z) = -\mathrm{E}(\log p(Y|Z))$, where the expectation is over the joint distribution of $Y$ and $Z$;
thus, unlike the conditional expectation $\mathrm{E}(Y|Z)$, the conditional entropy $H(Y|Z)$ is not a random variable.

By standard properties of entropy, since $f(\theta)$ is a function of $\theta$, the posterior entropy is
$$ H(\theta\mid X_{1:N}) = H(\theta \mid f(\theta), X_{1:N}) + H(f(\theta) \mid X_{1:N}). $$
By Condition \ref{condition:cond-indep}, $H(\theta \mid f(\theta), X_{1:N}) = H(\theta \mid f(\theta))$.
Further, $H(\theta \mid f(\theta)) = H(\theta) - H(f(\theta))$ again using that $f(\theta)$ is a function of $\theta$.
By Conditions~\ref{condition:identifiable} and \ref{condition:finite}, we have $H(f(\theta)\mid X_{1:N})\to 0$ as $N\to\infty$, because the posterior on $f(\theta)$ is guaranteed to concentrate at a single value \citep{Doob_1949,miller2018detailed}; see Lemma~\ref{lemma:doob-entropy} for details.
Thus, we have the following result.

\begin{theorem}
\label{theorem:general}
If $(\theta,\nu)\sim\pi$, $X_1,\ldots,X_N|\theta,\nu\sim P_{\theta,\nu}$ i.i.d., and $f(\theta)$
satisfies Condition~\ref{condition:general}, then
\begin{equation}
    \label{eqn:entropy-approx}
H(\theta\mid X_{1:N}) \xrightarrow[N\to\infty]{} H(\theta) - H(f(\theta)).
\end{equation}
\end{theorem}
In other words, when $N$ is sufficiently large, the difference between the prior entropy $H(\theta)$ and the posterior entropy $H(\theta|X_{1:N})$ is approximately equal to $H(f(\theta))$.
Put another way, the information gained---in terms of the reduction in entropy---is approximately equal to the entropy of the answer $f(\theta)$ \textit{under the prior}.
Thus, to estimate the information to be gained by a particular question $f$, we need only work with the prior --- not the posterior $\theta|X_{1:N}$ for a yet unobserved dataset $X_{1:N}$.

\subsubsection{Selection of experiments using approximate information gain}
\label{sec:selection-of-experiments}

To apply this result to select experiments, suppose that instead of $\pi(\theta,\nu)$ being the prior, $\pi(\theta,\nu)$ is the current posterior given all the data from any previous experiments. Let $\mathcal{E}$ denote a set of possible experiments. 
For each experiment $e\in\mathcal{E}$, let $X_1^e,\ldots,X_N^e|\theta,\nu \sim P_{\theta,\nu}^e$ i.i.d.\ be hypothetical random data from $N$ replicates of experiment $e$.
Suppose $f_e(\theta)$ satisfies Condition~\ref{condition:general} above.

Then $p(\theta \mid X_{1:N}^e) \propto p(X_{1:N}^e \mid \theta)\,\pi(\theta)$ is the posterior distribution of $\theta$ given the new data $X_{1:N}^e$ as well as data from any previous experiments. The expected posterior entropy of $\theta$ after experiment $e$ is then 
\begin{align}
\label{eqn:expected-posterior-entropy}
    H(\theta \mid X_{1:N}^e) = \int H(\theta \mid X_{1:N}^e = x_{1:N}^e)\, p(x_{1:N}^e)\, d x_{1:N}^e
    %\mathrm{E}\Big( -\sum_{\theta} p(\theta \mid X_{1:N}^e) \log p(\theta \mid X_{1:N}^e) \Big).
\end{align}
where $p(x_{1:N}^e)$ is the posterior predictive distribution for experiment $e$ given the data from previous experiments.
%Define $h_e = H(\theta \mid X_{1:N}^e)$ under model $e$. 
We would like to choose $e$ to minimize the expected posterior entropy $H(\theta \mid X_{1:N}^e)$.  However, approximating $H(\theta \mid X_{1:N}^e)$ via Monte Carlo is computationally intensive, since for every $e$, it would typically involve 
(i) simulating $T$ hypothetical datasets $x_{1:N}^{e,1},\ldots,x_{1:N}^{e,T}$ from the posterior predictive $p(x_{1:N}^e)$,
(ii) approximating the resulting new posteriors, $p(\theta \mid X_{1:N}^e = x_{1:N}^{e,t})$ for $t = 1,\ldots,T$, for example, by running $T$ MCMC chains,
and (iii) approximating the entropy $H(\theta \mid X_{1:N}^e = x_{1:N}^{e,t})$ for $t = 1,\ldots,T$,
in order to form a Monte Carlo approximation $H(\theta \mid X_{1:N}^e)  \approx \frac{1}{T}\sum_{t=1}^T \widehat{H}(\theta \mid X_{1:N}^e = x_{1:N}^{e,t})$ using Equation~\ref{eqn:expected-posterior-entropy}.

In contrast, the computation is vastly simplified using the approximation in Equation \ref{eqn:entropy-approx}.  First, note that by Equation \ref{eqn:entropy-approx}, 
\begin{align}
    \label{eqn:expected-entropy-approx}
    H(\theta \mid X_{1:N}^e) &\approx H(\theta) - H(f_e(\theta)). 
 %    &\approx \frac{1}{T} \sum_{t=1}^T \big[ H(\theta) - H(f_e(\theta)) \big] = H(\theta) - H(f_e(\theta)).
\end{align}
Since $H(\theta)$ does not depend on $e$, this implies that minimizing $H(\theta \mid X_{1:N}^e)$ is approximately equivalent to maximizing $H(f_e(\theta))$. Further, since $H(f_e(\theta))$ depends only on $\pi$ and $f_e$ (and not on $P_{\theta,\nu}^e$ or $X_{1:N}^e)$,
it is often relatively easy to approximate $H(f_e(\theta))$ using posterior samples $\theta_1,\ldots,\theta_T$. Specifically, we can generate a single set of samples $\theta_1,\ldots,\theta_T$ from the current posterior $\pi$, and then for each potential experiment $e\in\mathcal{E}$, compute
\begin{align} \label{eqn:approx-entropy-over-partition}
    H(f_e(\theta)) &=  -\sum_{y} p(f_e(\theta) = y) \log p(f_e(\theta) = y) \approx  -\sum_{y} \hat{p}_e(y) \log \hat{p}_e(y)
\end{align}
where $\hat{p}_e(y) :=  \frac{1}{T}\sum_{t=1}^{T}\mathds{1}\big(f_e(\theta_t) = y\big)$
and $\mathds{1}(\cdot)$ is the indicator function. 
In Equation~\ref{eqn:approx-entropy-over-partition}, the sum is over all values $y$ in the range of $f_e$. 
Thus, our proposed method of choosing $e$ is as follows.
\begin{enumerate}
    \item Generate samples $\theta_1,\ldots,\theta_T$ from the current posterior $\pi$.
    \item Compute $\hat{p}_e(y) =  \frac{1}{T}\sum_{t=1}^{T}\mathds{1}\big(f_e(\theta_t) = y\big)$ for each candidate experiment $e$.
    \item Select the experiment $e$ with the largest value of $\hat{h}_e := -\sum_{y} \hat{p}_e(y) \log \hat{p}_e(y)$. 
\end{enumerate}
Note that, equivalently, $\hat{h}_e = -\sum_{A\in\mathcal{A}_e} \hat{\pi}(A) \log \hat{\pi}(A)$ where $\hat\pi = \frac{1}{T}\sum_{t=1}^T \delta_{\theta_t}$ and $\mathcal{A}_e$
is the partition of $\theta$-space induced by $f_e$, since $\hat{p}_e(y) = \hat{\pi}(A)$ where $A = \{\theta : f_e(\theta) = y\}$. Therefore, we can interpret $\hat{h}_e$ as an approximation to the current posterior entropy of the partition induced by $f_e$.

It is important to note that although our criterion is motivated by the asymptotics as the number of replicates goes to $\infty$, it accounts for finite sample uncertainty due to the fact that the posterior $\pi$ quantifies our uncertainty in $\theta$ based on finitely many previous experiments and finitely many replicates of each previous experiment.
Also, a further advantage of our approach is that $f(\theta)$ often takes a small number of values, such as two for a binary function, and thus, $H(f(\theta))$ is often much easier to estimate from samples than $H(\theta \mid X_{1:N}^e)$ or even $H(\theta)$.

%I moved this to the causal network criterion section
% To apply this to causal networks, we could take $\theta = G$ and $f(G) = y$ where $y$ is such that $G\in A_y$ and $\mathcal{A} = (A_1,\ldots,A_K)$ is the partition of $\mathcal{G}$. It is possible that the assumed properties do not hold for all of the partitions we consider, but this could be justified by observing that a different choice of $f$ might be preferable when using finitely many replicates in practice.

\section{Criterion for causal network models}\label{sec:criterion-causal-network-models}
In this section, we apply our general criterion to the setting of causal network models.  First, we define the model we will use, and provide some intuition for partitions of graph space that are informed by interventional experiments.

\subsection{Causal network models}\label{subsec:causal-network-models}
We use the standard causal network model specification, which we review here. Note that it is common to refer to these models as ``Bayesian networks'' \citep{Pearl}, but we avoid this term because the Bayesian aspect of our methodology comes from the use of posterior distributions to quantify uncertainty, rather than from features inherent to the model itself.

All graphical models use nodes to represent random variables and edges to represent the probabilistic relationships among nodes. 
In contrast to traditional graphical models, which only specify the joint distribution in the observational setting, a causal network model also specifies the joint distribution when one or more nodes are manipulated.  
To represent this causal structure, it is standard to use a directed acyclic graph (DAG) along with the conditional probability distribution (CPD) of each node given the values of its parent nodes. The directed edges in this structure represent cause and effect relationships between parent and child nodes. We refer to a DAG topology along with all the CPDs as a \textit{causal network model}.

% If an intervention is applied that sets a node to a specific value, this may induce changes in the distributions of descendants of the manipulated node, but cannot affect the ancestors; thus, note that this operation is not the same as conditioning.

In a causal network model, the graph topology and the CPDs can be viewed as specifying an algorithm for generating data under manipulation of the nodes. Here, we assume interventions that assign some subset of nodes to values that may be fixed or random, but are independent of all other nodes. Thus, when intervening on node $i$, we effectively sever all incoming edges to node $i$. 
%(See \citet{Eaton07_DP} for a discussion of how to modify the marginal likelihood in the case of imperfect interventions or when the targets of an intervention are uncertain.) 
The data generating process under such an intervention can be described as follows: each manipulated node is set to its assigned value and each non-manipulated node is drawn from its CPD, proceeding in an ordering of the nodes that ensures parents are drawn before their children. In the observational setting (that is, when no nodes are manipulated), this reduces to the usual graphical model specification; that is, the joint distribution of the nodes $\mathcal{V} = (X_1, \dots, X_V)$ factors as
$p(X_1, \dots, X_V \mid \beta,\,G) = \prod_{i=1}^{V}p(X_i\mid X_{\mathrm{pa}(i)}, \beta_i,\,G)$
where $G$ is the graph topology and $\beta_i$ contains the parameters of the CPD for $X_i$.
However, the algorithmic nature of the causal network also specifies that, when some subset of nodes $S$ is independently manipulated, the joint distribution is
$p^*(X_1, \dots, X_V \mid \beta,\, G) = \prod_{i=1}^V p(X_i\mid X_{\mathrm{pa}(i)}, \beta_i,\, G)^{\mathds{1}(i\not\in S)} p^*(X_i\mid\beta_i^*)^{\mathds{1}(i\in S)}$
where $p^*(X_i\mid\beta_i^*)$ is the distribution of $X_i$ when intervening on node $i$.
Here, we define $\beta := (\beta_1,\ldots,\beta_V,\beta_1^*,\ldots,\beta_V^*)$.

For a data set $D = \big((X_{1,n}, \dots, X_{V,n}) : n=1,\ldots,N\big)$ consisting of $N$ samples of the $V$ nodes under interventions on subsets $S_1,\ldots,S_N$, respectively, the marginal likelihood is
\begin{align}\label{eqn:general-marginal-likelihood}
    p^*(D|G) &= \int p^*(D|\beta,G)\,p(\beta|G)\, d\beta \\ 
    &= \int \Big(\prod_{n=1}^N p^*(X_{1,n},\ldots,X_{V,n} \mid \beta,\,G)\Big)\,p(\beta|G)\, d\beta \\
    %\prod_{i=1}^{V} p^*(X_i^{1:N} \mid X_{\mathrm{pa}(i)}^{1:N},\, G) \\
    &= p_O^*(D|G) p_S^*(D)
\end{align}
where 
\begin{align}
p_O^*(D|G) &= \prod_{i=1}^{V} \int \big(\prod_{n=1}^{N} p(X_{i,n} \mid X_{\mathrm{pa}(i),n},\, \beta_i, \, G)^{\mathds{1}(i \not\in S_n)} \big)\, p(\beta_i | G)\, d\beta_i \\
p_S^*(D) &= \prod_{i=1}^{V} \int \Big(\prod_{n=1}^{N} p^*(X_{i,n}\mid \beta_i^*)^{\mathds{1}(i \in S_n)} \Big)\, p(\beta_i^*)\, d\beta_i^*.
\end{align}
When $p(\beta_i | G)$ and $p(\beta_i^*)$ are conjugate priors, these integrals can be computed in closed form. Since $p_S^*(D)$ does not provide any information about $G$, it is often omitted, however, we include it for theoretical purposes.

In this paper, we use categorical CPDs with Dirichlet priors, and thus the marginal likelihood $p^*(D\mid G)$ can be computed in closed form. Specifically, we assume that the CPD of each node is
\begin{align}
    p(X_i = k \mid X_{\mathrm{pa}(i)} = j,\; \beta,\, G) = \beta_{i j k} 
\end{align}
for $i\in\{1,\ldots,V\}$, $j\in\{1,\ldots,q_i\}$, and $k\in\{1,\ldots,r_i\}$.
Here, $j$ enumerates the possible joint states of $X_{\mathrm{pa}(i)}$, 
and we abuse notation slightly by writing $X_{\mathrm{pa}(i)} = j$ to mean that $X_{\mathrm{pa}(i)}$ takes the $j$th possible state. We use the BDeu Dirichlet prior, $\beta_{i j} \sim \mathrm{Dirichlet}(\boldsymbol{\alpha}_{i j})$ with
$\alpha_{i j k} = 1/(r_i q_i)$,
%where $r_i = |\mathcal{X}_i|$ (the number of states node $X_i$ can take) and $q_i = \prod_{j \in \mathrm{pa}(i)} |\mathcal{X}_j|$ (the number of joint states the parents of $X_i$ can take),
following standard practice in the categorical setting \citep{Heckerman1995, CooperYoo1999, Eaton07_hybridMCMC}.
Similarly, for the interventions, we assume $p^*(X_i = k \mid \beta_i^*) = \beta_{i k}^*$ and 
for simplicity, $\beta_i^* \sim \mathrm{Dirichlet}(1/r_i,\ldots,1/r_i)$. 

The BDeu  prior has the favorable property of likelihood equivalence, which we use in Section \ref{sec:markov-tse-equivalent} \citep{Buntine1991,Heckerman1995}. For the Dirichlet-Categorical case with BDeu prior,
\begin{align} \label{eqn:marginal-likelihood}
   p_O^*(D \mid G) &=  \prod_{i=1}^{V} \prod_{j=1}^{q_i} \frac{\Gamma(\alpha_{i j})}{\Gamma(\alpha_{i j}+N_{i j})} \prod_{k=1}^{r_i} \frac{\Gamma(\alpha_{i j k}+N_{i j k})}{\Gamma(\alpha_{i j k})} 
\end{align}
where $\Gamma$ is the gamma function, $N_{i j k} = \sum_{n=1}^N \mathds{1}(i\not\in S_n,\, X_{\mathrm{pa}(i)}=j,\, X_{i,n}=k)$ is the number of samples in which node $X_i$ is observed (not manipulated) to have state $k$ when its parents have state $j$, $N_{i j} = \sum_{k=1}^{r_i} N_{i j k}$, and $\alpha_{i j} = \sum_{k=1}^{r_i} \alpha_{i j k}$.
The OED method we propose can be used with other CPDs as well, so long as the marginal likelihood can be computed or approximated.

\subsection{Intuition for partitions of graph space}
\label{sec:graph-motivation}

In the Twenty Questions example, we considered the partition induced over a set of objects by asking a question about the object of interest. In this section, to provide intuition for how this applies to the causal network setting, we illustrate examples of partitions that intervention experiments induce over a space of graphs. 
For expository purposes, suppose the true graph $G$ is known to be one of the four graphs shown in Figure~\ref{fig:perturb}; this example is inspired by an example from \cite{Pournara}. 
In general, the set of graphs under consideration, $\mathcal{G}$, would consist of all possible DAGs on nodes $\{A,B,C,D\}$ rather than just these four graphs.

First, consider what features of a graph a single node intervention on node $e$ could help reveal. For example, we can expect an intervention on $A$ to have observable downstream effects on at least some of the descendants of $A$, but it would not affect any ancestors of $A$.
Therefore, intervening on $A$ should give us information about which nodes are descendants of $A$.
Thus, one possible choice for $f_e(G)$ is the set of descendants of the manipulated node.
Since $f_e(G)$ induces a partition of the set of graphs, we refer to it as a partition scheme. 
%We consider other features of a graph, such as children or parents of a manipulated node, that could be used for creating alternate partition schemes.

To select which node to manipulate, we compare the information each candidate intervention is expected to yield
with respect to a given partition scheme.
Suppose $A$ is manipulated. Figure \ref{fig:perturb} shows the partition of graphs according to the descendants of $A$; that is, $G$ and $G'$ are in the same part if $f_e(G) = f_e(G')$. In the first three graphs, $A$ has no descendants, whereas in the last graph, $\{B, C, D\}$ are all descendants of $A$. Thus, as long as intervening on $A$ has an effect on one or more descendants, we could distinguish whether $G \in \{G_1, G_2, G_3\}$ or $G = G_4$ after sufficiently many replicates of the intervention on $A$.
Meanwhile, if node $C$ is manipulated, a different partition over $\mathcal{G}$ is induced since $C$ has a different pattern of descendant sets than $A$; see Figure \ref{fig:perturb}.  

In general, an intervention partitions the set of graphs into equivalence classes such that (i) the graphs in each class are indistinguishable with respect to this intervention (corresponding to Condition~\ref{condition:cond-indep}), and (ii) graphs in different classes are distinguishable (corresponding to Condition~\ref{condition:identifiable}).
For the Dirichlet-Categorical model that we use, the equivalence classes induced by the likelihood have an elegant graph-based characterization; see Section \ref{sec:markov-tse-equivalent}.  However, we have also found several other partition schemes to be useful in practice; see Section \ref{sec:diffPartitions}.

After an intervention is performed, the generated data provide evidence to suggest which parts of the partition are compatible with the experimental data --- specifically, parts that are more compatible with the data will have higher posterior mass. 
%In the limit of infinite data, if the parts are identifiable under the given intervention, then all posterior mass will concentrate on the part that contains the true graph $G$.
Roughly speaking, we would like to choose an experiment that narrows down the set of compatible graphs as much as possible.
%, or in other words, an experiment that partitions $\mathcal{G}$ into a large number of small sets. 
For instance, in the toy example in Figure~\ref{fig:perturb}, intervening on $C$ is preferable to intervening on $A$, since $C$ induces a finer partition of $\mathcal{G}$. However, in general it is also important to consider the posterior probability of the graphs given the data from any previous experiments, since there is no point in finely partitioning regions of the space with very low probability.  To make this precise, we apply our general entropy-based criterion from Section \ref{sec:general-criterion} to the causal network setting, as described next.

\begin{figure}
\centering
        \includegraphics[scale = 0.60]{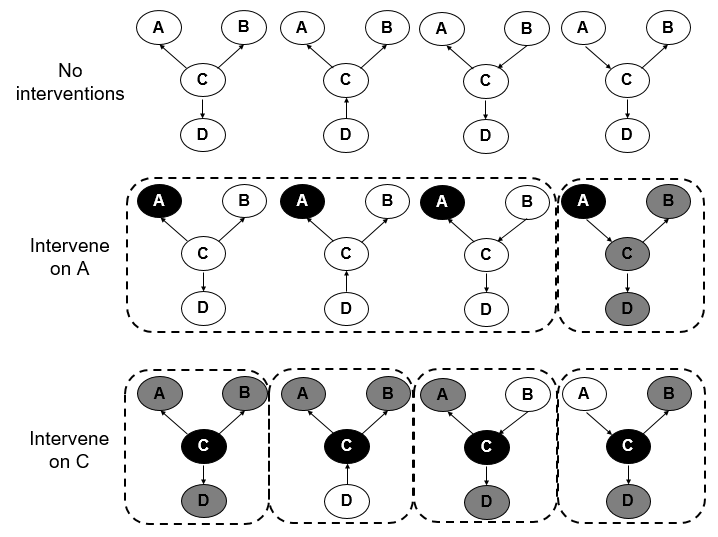}
    \caption{Top panel: Graphs in $\mathcal{G}$. Middle panel: Partition induced by intervening on $A$. The manipulated node is shown in black while descendants of the manipulated node are shaded grey. Bottom panel: Partition induced by intervening on $C$.}
    \label{fig:perturb}
\end{figure}

\subsection{Applying the experiment selection criterion to causal networks}
\label{sec:graph-criterion}

Specializing from the general setting of Section~\ref{sec:general-criterion} to the case of causal networks, we define the unknown parameter of interest to be $\theta := G$ and let $f_e(G)$ be a partition scheme.
%, and let $\mathcal{A}_e$ be the partition of $\mathcal{G}$ induced by $f_e$.  In other words, $\mathcal{A}_e$ is the collection of sets $\{G : f_e(G) = y\}$ for all possible values $y$ in the range of $f_e$.  
Further, in the Dirichlet-Categorical case, we define $\nu := \beta$, that is, the nuisance parameter $\nu$ is the collection of CPD parameters $\beta$.  Our goal is to perform experiments that make the posterior on graphs concentrate at the true graph as quickly as possible. We quantify concentration using the entropy of the posterior on graphs, $H(G)$, where $G$ is distributed according to the posterior given all experiments so far. Thus, we wish to perform experiments that minimize $H(G)$.

In many cases, inferring the entire graph $G$ is overly ambitious. For instance, if the number of nodes is even moderately large, the number of possible graphs is extremely large, making it infeasible to infer $G$ completely.  However, often, one only needs to infer a specific feature of $G$, such as whether node $X_1$ is an ancestor of node $X_2$, or whether $X_3$ mediates the effect of $X_1$ on $X_2$. In such cases, one can instead define $\theta$ to be a function of $G$, say, $\theta:=\varphi(G)$ and focus on minimizing $H(\varphi(G))$ rather than $H(G)$.

%As discussed in Section \ref{sec:graph-motivation}, some experiments provide more information about the structure of $G$ than others. 
To prioritize experiments, we apply the general criterion derived in Section \ref{sec:general-criterion}.
Specifically, given samples $G_1, \dots, G_T$ from the current posterior, we choose the next experiment $e$ to maximize 
\begin{align}\label{eqn:hhat}
    \hat{h}_e = -\sum_y \hat{p}_e(y) \log \hat{p}_e(y)
\end{align}
where $\hat{p}_e(y) = \frac{1}{T}\sum_{t = 1}^T \mathds{1}(f_e(G_t) = y)$.
Thus, we choose the intervention that maximizes the posterior entropy (under the current posterior) of the partition induced by $f_e$.
If Condition~\ref{condition:general} is satisfied, then this minimizes the approximate expected entropy of the new posterior given the additional data from the experiment.
Meanwhile, if Condition~\ref{condition:general} is not fully satisfied, then this approach is not guaranteed to reduce the entropy optimally, but is still a sensible way of choosing interventions that quickly reduce the entropy.

\section{Practical implementation of the method}\label{sec:practical-implementation}

In this section, we provide practical details on implementing the criterion in Section \ref{sec:graph-criterion},
including specifics regarding partition schemes, equivalence classes of graphs, sampling from the posterior on graphs,
and an overall algorithm.

\subsection{Partition schemes}
\label{sec:diffPartitions}

In the context of graphs, Condition~\ref{condition:cond-indep} is that given $f_e(G)$, the graph $G$ is conditionally independent of data from experiment $e$, and Condition~\ref{condition:identifiable} is that $f_e(G)$ is identifiable with respect to the distribution of the data under experiment $e$. Condition~\ref{condition:finite} is always satisfied since there are finitely many graphs $G$.

While in theory we use Condition~\ref{condition:general} to justify the method, in practice Conditions \ref{condition:cond-indep} and \ref{condition:identifiable} are not strictly necessary.
Recall that the optimality theory is based on the expected information gain from asymptotically many replicates of the next selected experiment.
Thus, in practice, it may be possible to obtain excellent performance using a partition scheme that violates Condition \ref{condition:cond-indep} or \ref{condition:identifiable}.
Consequently, we define a variety of partition schemes here, and we empirically compare their performance in Section~\ref{sec:simulations}.

We consider experiments $e$ that intervene on a single node, and for brevity we use $e$ to denote the manipulated node.  Consider the following partition schemes.
\begin{enumerate}
    \item Markov equivalence class (MEC): $f_e(G)$ equals the Markov equivalence class of graphs when intervening on node $e\in\mathcal{V}$; see Section \ref{sec:markov-tse-equivalent}.
    \item Child Set (CS): $f_e(G)$ equals the set of children of node $e\in\mathcal{V}$.
    \item Descendant Set (DS): $f_e(G)$ equals the set of descendants of node $e\in\mathcal{V}$.
    \item Parent Set (PS): $f_e(G)$ equals the set of parents of node $e\in\mathcal{V}$.
\end{enumerate}
We also consider the following slightly different approach.
\begin{enumerate}\setcounter{enumi}{4}
    \item Pairwise Child (PWC): Maximize $\sum_{v\in\mathcal{V}} H(f_{e,v}(G))$ where $f_{e,v}(G) = \mathds{1}((e,v)\in G)$ is the indicator of whether $G$ has an edge from $e$ to $v$.
    % Rather than partition according to the child sets of each manipulated node, for each node $v \in \mathcal{V} \setminus e$, we partition the space of graphs into two parts, according to whether or not there is an edge from $e$ to $v$ in the graph. We then seek to maximize the average pairwise entropy over all nodes, 
\end{enumerate}

\subsection{Markov equivalence classes}
\label{sec:markov-tse-equivalent}
Two DAGs are said to be \textit{Markov equivalent} if they represent the same set of conditional independence relations.  While all of the conditional independence relationships entailed by a graph can be computed using the d-separation algorithm \citep{Pearl88}, the following elegant result provides a simpler way to determine whether two DAGs are Markov equivalent based on their topology.

\begin{theorem}[\cite{Verma91}] \label{theorem-markov-equiv}
Two DAGs $G_1$ and $G_2$ are Markov equivalent if and only if they have the same skeleton and the same v-stuctures. 
\end{theorem}

The \textit{skeleton} of a graph is its topology ignoring edge directions. A \textit{v-structure} is a triple of nodes $(x,y,z)$ with topology $x \rightarrow y \leftarrow z$, where there is no edge connecting $x$ and $z$. 
%A Markov equivalence class can be represented with a completed partially directed acyclic graph (CPDAG) such that (a) the CPDAG has the same skeleton as all graphs in the class, and (b) each edge in the CPDAG is directed if and only if that edge has the same direction in all graphs in the class; otherwise that edge is undirected in the CPDAG \citep{Chickering:UAI96}. 
In general, observational data alone cannot distinguish between Markov equivalent graphs unless one assumes specific error distributions or functional model classes \citep{Peters2011}. A rich literature exists on Markov equivalence; see, for example, \cite{Andersson1997} and \cite{Chickering:UAI96}.

While Markov equivalent graphs represent the same conditional independence relationships, they differ in the causal relationships they encode since it is the direction of arrows, not just the skeleton, that is important for causal interpretation. Interventions can help distinguish among graphs in the same Markov equivalence class. However, even after an intervention is performed, some graphs may still be indistinguishable. \cite{Hauser2012} consider performing a sequence of interventions, and they provide a generalization of Theorem~\ref{theorem-markov-equiv} that characterizes the equivalence classes of graphs that are indistinguishable with respect to the whole sequence of interventions.
% refer to these indistinguishable graphs as interventionally Markov equivalent. 

For our approach, however, we only need to consider the partition induced by a single candidate intervention (rather than the whole sequence of interventions), since the information from previous interventions is already represented in the posterior distribution.  Thus, for our approach, a natural choice of partition scheme is to define $f_e(G)$ to be the Markov equivalence class of $G^e$, where $G^e$ is the DAG obtained from $G$ by removing all edges from $\mathrm{pa}(e)$ to $e$. This is referred to as the ``MEC'' scheme in Section~\ref{sec:diffPartitions}.

Following the notation of Sections~\ref{sec:general-criterion} and \ref{sec:criterion-causal-network-models},
we write $P_{G,\beta}$ for the distribution of $(X_1,\ldots,X_V)$ given graph $G$ and CPD parameters $\beta = (\beta_1,\ldots,\beta_V,\beta_1^*,\ldots,\beta_V^*)$.
Define $\beta^e$ to be a modified copy of $\beta$ in which $\beta_e^*$ takes the place of $\beta_e$ and $\beta_{e 1}$ takes the place of $\beta_e^*$.
% that is obtained by copying $\beta$ and then (i) replacing $\beta_e$ by $\beta_e^*$, and (ii) replacing $\beta_e^*$ by $\beta_{e 1}$.
Thus, when intervening on node $e$, the distribution can be written as $P_{G^e,\beta^e}$.
%Let us write $X^e = (X_1^e,\ldots,X_V^e)$ to denote a random vector such that $X^e \sim P_{G^e,\beta^e}$.

Consider a sequence of interventions in which a single node is manipulated at a time.
For instance, suppose we have performed $N_k$ replicates intervening on node $i_k$ for $k = 1,\ldots,K$,
and we are considering intervening on node $i'$ for the next set of $N'$ replicates.
The joint model is then
\begin{align} \label{eqn:sequential-model}
\begin{split}
    & (G,\beta) \sim \pi \\
    & X_1^e,\ldots,X_N^e \mid G,\beta \text{ i.i.d.} \sim P_{G^e,\beta^e} \text{ for $(e,N) \in \{(i_1,N_1),\ldots,(i_K,N_K),(i',N')\}$}
\end{split}
\end{align}
where $P_{G,\beta}$ is the categorical model, $\pi(\beta|G)$ is the BDeu-based prior defined in Section~\ref{subsec:causal-network-models}, and $\pi(G)$ is an arbitrary prior on DAGs.
\begin{theorem}\label{theorem:cond-indep}
Under the joint model in Equation~\ref{eqn:sequential-model}, if $f_e(G)$ is the Markov equivalence class of $G^e$ then
$$ X_{1:N'}^{i'} \indep G \mid f_{i'}(G),f_{i_1}(G),\ldots,f_{i_K}(G). $$
\end{theorem}

\begin{theorem}\label{theorem:identifiable}
Assume the joint model in Equation~\ref{eqn:sequential-model}, and let $D = (X_{1:N_1}^{i_1},\ldots,X_{1:N_K}^{i_K})$ denote the data observed so far.
If $f_e(G)$ is the Markov equivalence class of $G^e$ then
there is a function $g$ such that $g(P_{G^{i'},\beta^{i'}}) = f_{i'}(G)$ almost surely when $(G,\beta)\sim p(G,\beta\mid D)$.
\end{theorem}

Now, to employ Theorem~\ref{theorem:general}, observe that under the model in Equation~\ref{eqn:sequential-model},
if we condition on $D$ then we obtain the following model:
\begin{align*}
    & (G,\beta) \sim p(G,\beta \mid D) \\
    & X_1^{i'},\ldots,X_N^{i'} | G,\beta \text{ i.i.d. } \sim P_{G^{i'},\beta^{i'}}.
\end{align*}
This follows the form of the abstract model in Section~\ref{sec:general-criterion}, with appropriate notational substitutions.
As above, let $f_e(G)$ be the Markov equivalence class of $G^e$.
Condition~\ref{condition:cond-indep} is that $X_{1:N'}^{i'} \indep G \mid f_{i'}(G)$ in this conditional model given $D$,
or equivalently, $X_{1:N'}^{i'}\indep G \mid f_{i'}(G),D$ under the joint model in Equation~\ref{eqn:sequential-model}.
By Theorem~\ref{theorem:cond-indep}, $X_{1:N'}^{i'} \indep G \mid f_{i'}(G),f_{i_1}(G),\ldots,f_{i_K}(G)$,
so we can expect that $X_{1:N'}^{i'}\indep G \mid f_{i'}(G),D$ holds approximately when $N_1,\ldots,N_K$ are sufficiently large,
since $D = (X_{1:N_1}^{i_1},\ldots,X_{1:N_K}^{i_K})$ and $X_{1:N_k}^{i_k}$ pertains to $f_{i_k}(G)$.
Condition~\ref{condition:identifiable} is that there exists $g$ such that $g(P_{G^{i'},\beta^{i'}}) = f_{i'}(G)$ almost surely
when $(G,\beta)\sim p(G,\beta\mid D)$, which is precisely what Theorem~\ref{theorem:identifiable} shows.
Finally, Condition~\ref{condition:finite} is that $f_{i'}(G)$ takes finitely many values, which is true since there are only finitely many graphs $G$ on $V$ nodes.

Therefore, Theorem~\ref{theorem:general} indicates that selecting the next intervention using the strategy in Section~\ref{sec:selection-of-experiments} with $f_e(G)$ chosen to be the Markov equivalence class of $G^e$ is a natural choice 
to optimally reduce entropy, under the asymptotic approximation that the number of replicates in each experiment is sufficiently large.

\subsection{Sampling from the posterior distribution on graphs}
\label{sec:sampleDAGs}

A large body of work exists on MCMC methods for sampling from $p(G|D)$. This is a challenging task since the number of DAGs increases super-exponentially with the number of nodes and the posterior on graphs is often highly multi-modal. Some have proposed searching the space of graphs using local proposals that add, delete, or reverse edges at random \citep{Madigan1995} and others have improved chain mixing by sampling over the space of node orderings \citep{Friedman2003, Ellis2006}.  We use a clever MCMC algorithm developed by \citet{Eaton07_hybridMCMC} that uses dynamic programming (DP) to construct proposals.  This method explores the space of DAGs using a Metropolis-Hastings algorithm with a proposal distribution that is a mixture of local moves (edge deletions, additions, or reversals) and a global move that proposes a new graph in which an edge exists between two nodes with probability equal to the exact marginal posterior edge probability, computed using DP.

Key to the DP algorithm's ability to compute exact marginal posterior edge probabilities is the assumption of a ``modular prior" over structures. Rather than directly specifying a prior over DAGs, a modular prior requires specifying a prior over node orderings and a prior that gives weight to sets of parents (and not to their relative order). Together these terms define a joint prior over graphs and orders. Defining the prior in this way allows the contribution to the marginal likelihood for nodes with the same parent sets to be cached and re-used for efficient exact computation, regardless of the orderings of the parents; see \citet{Koivisto2006} for details of the DP algorithm and see \citet{Koivisto04}, \citet{Friedman2003}, and \citet{Ellis2006} for further discussion of priors on orderings and graphs. 

A modular prior tends to favor graphs that are consistent with more orderings, such as fully disconnected graphs and tree structures. 
%both of which are sparse graphs with small parent sets. 
In fact, the modular prior favors tree structures over chains even if the two structures are Markov equivalent. For instance, tree structure $1 \leftarrow 2 \rightarrow 3$ has higher prior probability than the chain $1 \rightarrow 2 \rightarrow 3$ under a modular prior since the tree structure is consistent with two node orderings \citep{Eaton07_hybridMCMC}. While one may want to use a uniform prior over DAGs in the absence of prior knowledge, \cite{Ellis2006} and \cite{Eaton07_hybridMCMC} show how a uniform prior over orderings and flat prior over parent sets together encode a highly nonuniform prior over DAGs. The hybrid MCMC-DP approach that we use \citep{Eaton07_hybridMCMC}  overcomes this limitation of the DP algorithm. With MCMC-DP, we can use an arbitrary prior on graphs and draw valid samples from $p(G|D)$, while benefiting from a fast, data-driven proposal distribution to help traverse the DAG space.  We implemented our method in MATLAB (version 2017a) and we use the BDAGL package \citep{Eaton07_hybridMCMC} to sample from $p(G|D)$ using the MCMC-DP algorithm. Source code is available online at https://github.com/mzemplenyi/OED-graphical-models.

\subsection{Overall algorithm}

The inputs to our proposed algorithm are (i) a set of candidate experiments $\mathcal{E}$, (ii) a mechanism for generating i.i.d.\ samples $(X_1,\ldots,X_V)$ from $P_{e,0}$, the true distribution under experiment $e\in\mathcal{E}$, and (iii) a partition scheme $f_e(G)$ for each experiment $e\in\mathcal{E}$.  For the first experiment, we generate observational data by not intervening on any nodes. Each subsequent experiment sets a single node from $\mathcal{E}$ to a fixed value.  
%We focus on experiments that set some node to a particular value, or the ``control'' experiment that generates observational data by not intervening.
The algorithm proceeds as follows.  Let $D$ denote the collection of data from the experiments so far.
\begin{enumerate}
    \item Obtain posterior samples from $\pi(G) = p(G|D)$ and approximate the posterior entropy, $H(G)$.
    \item Check the stop criteria. Stop the algorithm if either:
    \begin{enumerate}
        \item $H(G)$ falls below a given entropy tolerance threshold, or 
        \item the maximum number of allowed experiments has been reached.
    \end{enumerate}
    Otherwise, continue.
    \item\label{item:enumpartition} For each $e \in \mathcal{E}$, enumerate the partition \textit{over the sampled graphs} induced by $e$ and calculate $\hat{h}_e$, the approximate posterior entropy over the partition (Equation~\ref{eqn:hhat}).
    \item Select the experiment $e$ that maximizes $\hat{h}_e$ as the next intervention experiment to perform. 
    \item Generate data for experiment $e$: draw $N$ i.i.d.\ samples from $P_{e,0}$.
    \item Combine the new data with the existing data and repeat from the beginning.
\end{enumerate}
For computational efficiency, note that in step \ref{item:enumpartition}, it is only necessary to consider those parts $A$ in the partition $\mathcal{A}_e$ that contain one or more posterior samples.  Since it is not necessary to consider all parts in the partition, this can provide a computational advantage in cases where there are an intractable number of parts.

\section{Previous work}\label{sec:previous-work}
\label{sec:methodsreview}
Previously proposed methods for learning causal network models from observational and interventional data in the non-OED setting typically fall into one of two categories: (1) constraint-based methods, such as the PC-algorithm \citep{Spirtes2001}, that test for conditional independence constraints in the data and select models that match those constraints, and (2) score-based methods, wherein the space of structures is searched for ones that that are most supported by the data, as quantified by scores such as the Bayesian marginal likelihood or BIC. More recently, OED methods, also referred to as active learning methods, have been developed for both approaches.  

\cite{HeGeng2008}, \cite{Eberhardt08}, and \cite{Hauser2012b:EuroPGM} propose constraint-based active learning methods that first use observational data or prior knowledge to construct a partially directed acyclic graph (PDAG), also called an essential chain graph. They then use graph-theoretic results to select the interventions required to orient all edges in the essential chain graph, using variations on criteria that generally seek to minimize the number of undirected edges in the post-intervention equivalence class of graphs. These algorithms take as a starting point a known observational essential graph, which would require infinite observational data in principle, and in practice is estimated from finite observational data. However, they often do not perform as well in finite sample settings where estimation errors introduced in the initial chain graph can lead to an incorrectly estimated DAG.  \cite{Hauser2012} demonstrate that, in the finite sample setting, estimation errors can be reduced by using interventional data to refine not only the directionality of uncompelled edges in the chain graph (as done by He and Geng), but also the skeleton of the chain graph.  

\cite{TongKoller} and \cite{Murphy2001} develop score-based active learning methods for Bayesian networks. They use MCMC to sample from the space of node orderings (Tong and Koller) or graphs (Murphy) along with a decision-theoretic framework to select interventions. Both methods are computationally expensive since they involve computing the predictive density for each sampled graph subject to each possible intervention. Other methods forgo the predictive sampling step and instead use selection criteria with lower computational burdens. \cite{Pournara} consider the equivalence classes of high-scoring networks (determined by a greedy hill-climbing search) and select interventions that tend to partition transition sequence equivalence classes into smaller and smaller subclasses.  \cite{LiLeong} propose a `non-symmetrical entropy' criterion closely related to Tong and Koller's loss function, but use the DP algorithm rather than MCMC to calculate edge probabilities between nodes. \cite{Cho} adapt the active learning framework of Murphy (2001) to the Gaussian Bayesian network setting. \cite{Ness} use a Bayesian framework that allows one to directly encode prior causal knowledge about each edge, which then induces a prior over graphs. Their method, bninfo, takes a set of highly scoring PDAGs and returns the minimally-sized batch of interventions that is expected to correctly orient the greatest number of edges; this method shares elements of both the score-based and constraint-based methods. 

The previous work that is most similar to ours is the method of \citet{Almudevar}, who explore the performance of an entropy-based criterion that uses the idea of partitioning the space of graphs to minimize the entropy on the posterior on graphs.  While our method is based on a similar theoretical justification as that of \citet{Almudevar}, we generalize the approach to a large class of partition schemes and we improve performance by using a more efficient MCMC procedure. Further, in contrast to the limited simulation study of \citet{Almudevar}, we provide a more extensive set of empirical results on a wide variety of networks, we compare with other leading algorithms, and we make our software publicly available.

For a comprehensive review of optimal experimental design methods for other models, including Boolean networks and  differential equation models, see \citet{Sverchkov}.

\section{Simulation results}\label{sec:simulation-results}
\label{sec:simulations}

In this section we first evaluate the performance of our OED method under various partition schemes. Then we compare our method to the methods of \citet{LiLeong} and \citet{Ness}. Throughout, we evaluate performance using two metrics:
\begin{enumerate}
    \item Mean Hamming distance: After each experiment, we calculate the posterior probability of an edge between nodes $X_i$ and $X_j$ via:
    \begin{align}
        p(i \rightarrow j\mid D) = \sum_{G \in \mathcal{S}_{i \rightarrow j}} p(G|D)
    \end{align}
    where $\mathcal{S}_{i \rightarrow j}$  is the set of graphs containing the edge $i \rightarrow j$. We then construct the median probability graph, defined as the graph containing only those edges for which $p(i \rightarrow j\mid D) \geq 0.5$ \citep{Castelletti2018, Peterson2015}. The Hamming distance between the median probability graph and the ground truth network is equal to the number of false detected edges (false positives) and missing edges (false negatives) in the median probability graph. We then take the average of this distance over all simulations.  
 
    \item Mean true positive rate (TPR): After each experiment, we calculate the proportion of correctly detected edges present in the median probability graph among the edges in the ground truth network. We then find the average of this proportion over all simulations. 

\end{enumerate}
For all figures, error bars represent the standard error of the mean over 50 simulations.    
% For the various ground truth networks that we considered for our simulation studies, we ran $n_{sims} = 50$ simulations where in each simulation consisted of a series of $n_{exp}$ experiments. For the first experiment we generated $n_{obs}$ observational samples and for each subsequent experiment we generated $n_{pert}$ pertubation samples, with the perturbed node varying by active learning method. Between each experiment, we used MCMC to draw 250,000 posterior samples from $P(G|D)$ and discarded the first 150,000. We generated a distinct set of observational and pertubation data for each of the 50 simulations, but within a simulation, all methods used the same data (i.e. the 1000 observational samples used in the first experiment of the first simulation were the same for all six methods). 

\subsection{Comparison of partition schemes}
Our first simulation study compared the five different ways, defined in Section \ref{sec:diffPartitions}, to partition graphs sampled from the posterior $p(G|D)$.  We randomly generated a discrete graph with 10 binary nodes and generated observational and intervention samples from this ground truth network; see Figure \ref{fig:graph-structures} for the graph's structure.
% made that network via bnlearn package: 
% ten2 <- random.graph(LETTERS[1:10], method = "ic-dag", max.in.degree = 2).
\begin{figure}
	\centering
	\begin{tabular}{cc}
  \includegraphics[scale = 0.6]{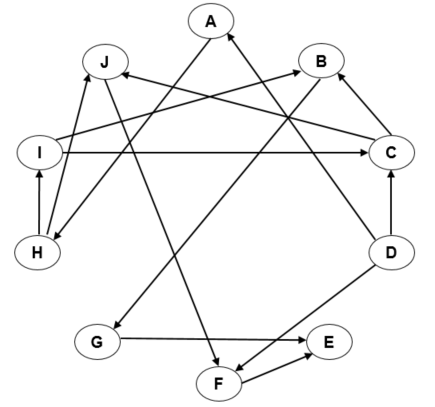}   &  \hspace{1cm}
  \includegraphics[scale = 0.6]{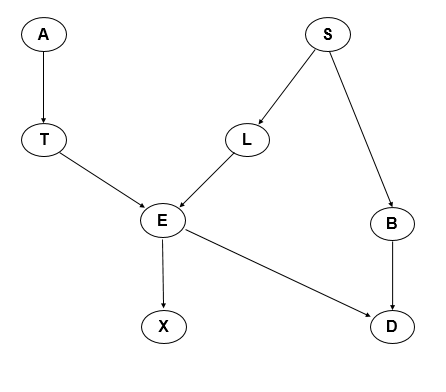}
\end{tabular}
\caption{Left: Structure of ten-node network. Right: Structure of Asia network as defined by the Bayesian Network Repository available in the \emph{bnlearn} R package \citep{Scutari}.}
	\label{fig:graph-structures}
\end{figure}

For each partition scheme, we ran $n_{sims} = 50$ simulations wherein each simulation consisted of a series of $n_{exp} = 7$ experiments. For the first experiment, we generated $n_{obs} = 1000$ observational samples to form the initial dataset $D$. For each of the six subsequent experiments, we generated $n_{intv} = 1000$ additional intervention samples and appended them to $D$, where the manipulated node was selected via our entropy-based selection criterion based on the partition scheme under consideration. After generating data for each experiment, we used MCMC-DP to draw 250,000 posterior samples from $p(G|D)$, discarded the first 150,000 samples as burn-in, and used the remaining 100,000 samples for posterior inference. We used a uniform prior over graphs and did not allow a node to be manipulated more than once. In addition to the five partition schemes, we also evaluated a ``random learner'' that randomly selected the next node to be manipulated, rather than using an entropy criterion. 

Figure \ref{fig:internal} shows the mean Hamming distance and mean TPR for the five entropy-based methods and the random learner on the 10-node network. Each of the entropy-based methods performed better than the random learner according to both metrics. For this network, the entropy-based methods all performed similarly. We found this to be the case in simulations using other randomly generated networks as well.

The benefit of an entropy-driven experimental design method over randomly selected interventions is evident from Figure \ref{fig:internal}.  To fall within a given mean Hamming distance from the true graph required far fewer interventions using the OED methods compared to the random learner.
%As results in subsequent sections will show, the number of interventional experiments an entropy-driven method is expected to save depends on a number of factors including the size and structure of the data-generating network as well as the number of observational and interventional samples generated in each experiment. 
% The need to perform fewer experiments could translate into non-trivial time and cost savings.      

\begin{figure}
\centering
\includegraphics[scale = 0.65]{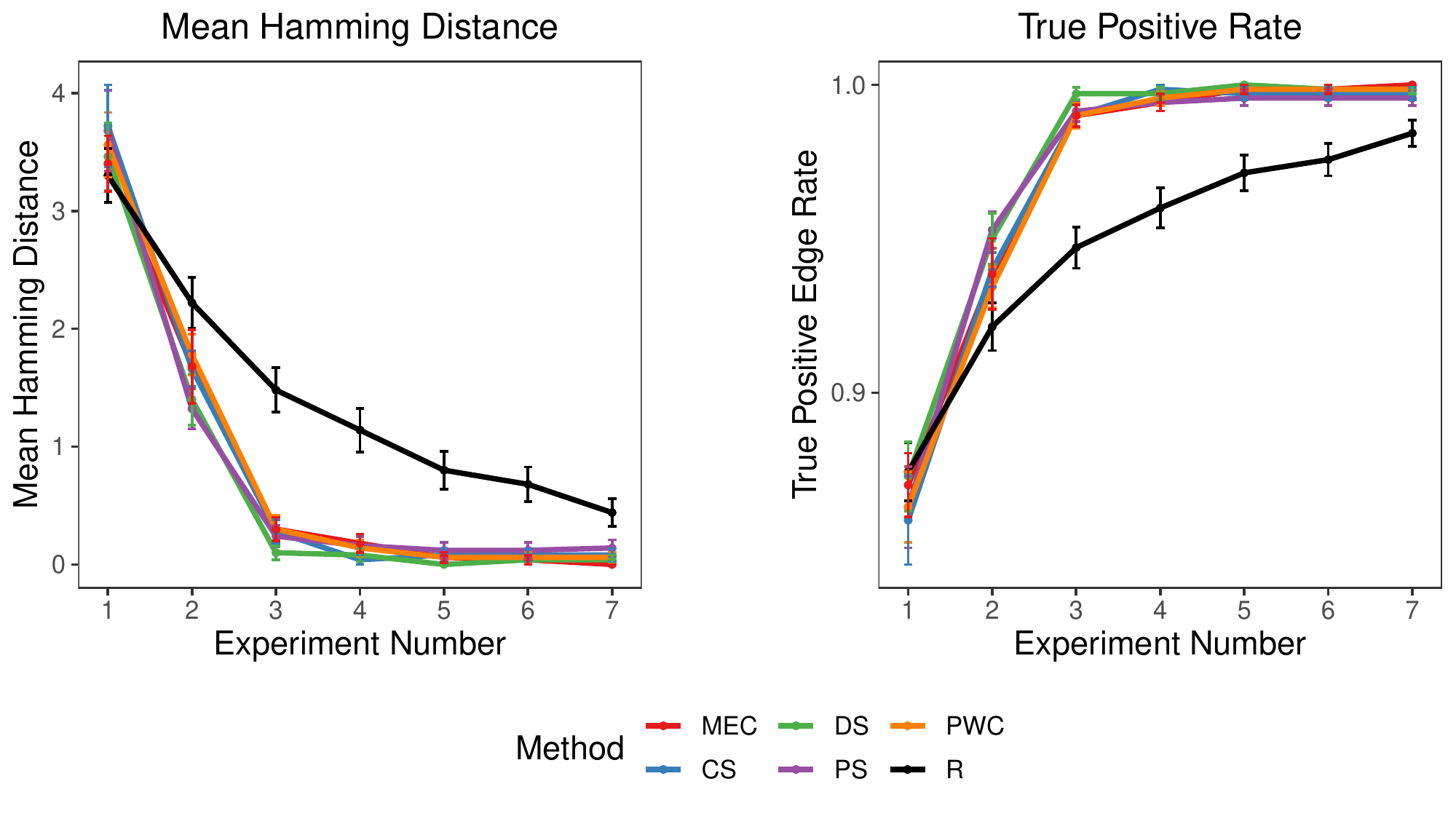}
	\caption{Mean Hamming distance and mean TPR for five entropy-based OED methods (each using a different partition scheme) and the random learner on a ten-node binary network. Settings: $n_{sim} = 50$, $n_{exp}$ = 7, $n_{obs} = 1000$, $n_{intv} = 1000$. MEC = Markov equivalence class; CS = child set; DS = descendant set; PS = parent set; PWC = pairwise child; R = random learner.}
	\label{fig:internal}
\end{figure}

\subsection{Comparison with other methods}
In addition to comparing various partition schemes, we also compared our OED algorithm to two other active learning methods. 

The first is a method proposed by \citet{LiLeong} that also uses an entropy-based criterion. In fact, what Li and Leong refer to as their non-symmetrical entropy criterion for selecting interventions is equivalent to the pairwise child (PWC) entropy criterion described in Section \ref{sec:diffPartitions}. The difference between the methods, however, is that \citet{LiLeong} do not use MCMC to sample from the posterior on graphs. Rather, they evaluate their criterion using exact edge probabilities computed using a DP algorithm \citep{Koivisto2006}; see Section~\ref{sec:sampleDAGs} for details on the DP algorithm and its assumptions. We refer to the method of Li and Leong as ``DP" in subsequent figures. 

The second method, ``bninfo"  \citep{Ness}, evaluates the expected causal information gain of candidate interventions and outputs a minimally-sized batch of interventions expected to maximize that gain. \citet{Ness} define causal information gain as the increase in correctly oriented edges in the causal network. Their algorithm constructs the recommended batch of interventions one node at a time, in descending order of expected causal information gain. In order to compare our method with bninfo, for a given causal network, we took the sequence of interventions that bninfo recommended and used MCMC-DP to sample from the posterior distribution on graphs between each recommended intervention. This allowed us to construct the median probability graph and calculate the Hamming distance and TPR after each intervention that bninfo recommended. Note that \citet{Ness} provide a way to encode prior knowledge on each edge in the graph, but to facilitate comparison with the other methods, we use a uniform prior on the graph topology. 

\subsubsection{Asia network}
 We first assessed performance of the various methods on the Asia network, a commonly used network for comparing network inference methods, first described by \citet{Lauritzen1998}. The Asia network consists of eight binary nodes that describe the relationship between lung diseases and visits to Asia (Figure \ref{fig:graph-structures}). We used the conditional probability table provided by the \emph{bnlearn} R package to generate observational and interventional data \citep{Scutari}. 
 
 Figure \ref{fig:asia} compares the performance of our method (using the MEC partition scheme), DP, bninfo, and the random learner. For this simulation study we used the following settings: $n_{sim} = 50$, $n_{exp} = 9$, $n_{obs} = 300$, $n_{intv} = 300$. For the MCMC methods, we drew 250,000 samples, discarded the first 150,000 samples, and used the remaining 100,000 samples for inference. For the sake of clarity, we omitted the other partition schemes since they performed similarly to the MEC partition scheme. 
 
After the first intervention experiment, the four methods performed nearly identically. The methods then diverged at experiment 3, with the random learner and bninfo lagging behind MEC and DP. Note that the maximum batch size of interventions recommended by bninfo across the simulations consisted of seven nodes, so the bninfo results in Figure \ref{fig:asia} only extend to the eighth experiment (the observational experiment followed by seven interventions). For the other methods, by the ninth experiment, all 8 nodes had been manipulated since we did not allow for repeat interventions.  Thus, it makes sense that the MEC and random learner results align by the last experiment; they have each sampled the same data, albeit in different orders. Even though the DP method had also sampled intervention data for all nodes by the ninth experiment, it does not converge with the other methods because the DP method uses a different prior over graphs (a non-uniform prior induced by its modular joint prior over node ordering and parent sets as described in Section \ref{sec:sampleDAGs}). 
 
\begin{figure}
	\centering
     \includegraphics[scale = 0.6]{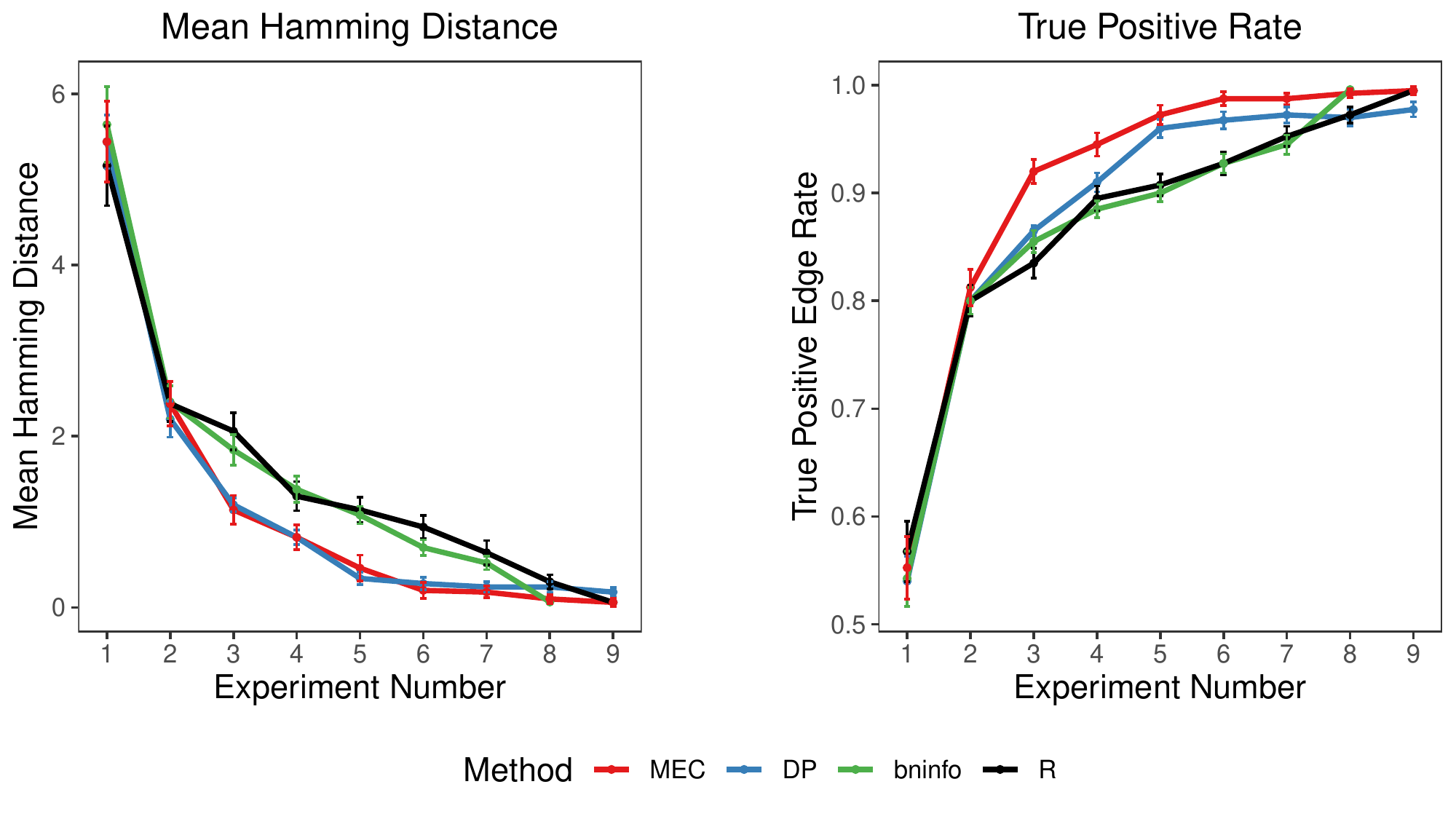}   
	\caption{Mean Hamming distance and mean TPR on the 8-node Asia network with the following settings: $n_{sim} = 50$, $n_{exp}$ = 9, $n_{obs} = 300$, $n_{intv} = 300$. MEC = our method with MEC partition scheme; DP = dynamic programming method of \citet{LiLeong}; R = random learner.}
	\label{fig:asia}
\end{figure} 

\subsubsection{Effect of network topology on inference}
Next, we explored the effect of network topology on performance of MEC, PWC, DP, bninfo, and the random learner. Here, as above, ``PWC'' refers to using our method with the PWC partition scheme. We include PWC for direct comparison with the DP method in order to illustrate how two methods that employ the same entropy criterion for selecting interventions---but differ in how posterior edge probabilities are calculated (see Section~\ref{sec:sampleDAGs})---may perform differently depending on network topology. We considered two 8-node networks (Figure \ref{fig:chain-tree-structures}), one with a chain structure and one with a tree structure. 

For the chain network, we used the settings $n_{sim} = 50$, $n_{exp}$ = 7, $n_{obs} = 1000$, $n_{intv} = 1000$. The DP algorithm performed worse than the other methods, including the random learner, until the fourth experiment. Interestingly, DP performed worse than PWC, even though the two use the same entropy criterion. This can be explained by the fact that, as described in Section \ref{sec:sampleDAGs}, the non-uniform prior over graphs that the DP algorithm uses in order to achieve its computational efficiency puts less mass on chain structures relative to other structures. Meanwhile, the hybrid MCMC-DP approach we used in the PWC method does not have such constraints on its prior over structures, thus, it performed better in this setting since a uniform prior on graphs was used.

For the tree network, we used the settings $n_{sim} = 50$, $n_{exp}$ = 8, $n_{obs} = 200$, $n_{intv} = 200$.  MEC outperformed the other methods according to both mean Hamming distance and TPR, but all methods performed well, falling within a mean Hamming distance of one from the true network by the third experiment. While the DP method initially had a higher mean Hamming distance and lower TPR than the other methods, by the fourth experiment DP outperformed the PWC method. This illustrates that the DP method can work better than PWC when the true graph is more probable under its implicitly assumed prior. 

\begin{figure}
\centering
\begin{tabular}{ll}
{\includegraphics[scale = 0.4]{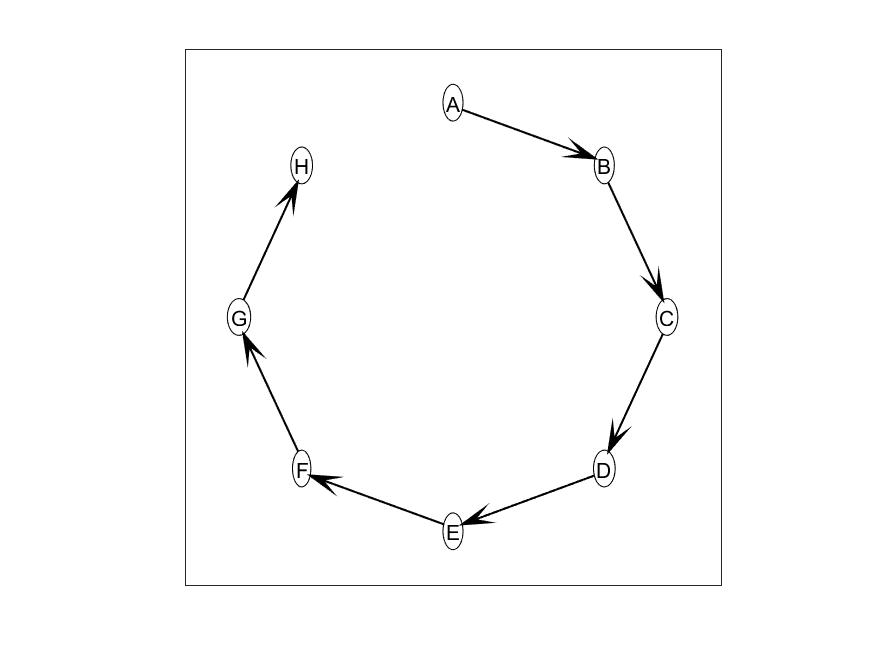}} &
{\includegraphics[scale = 0.4]{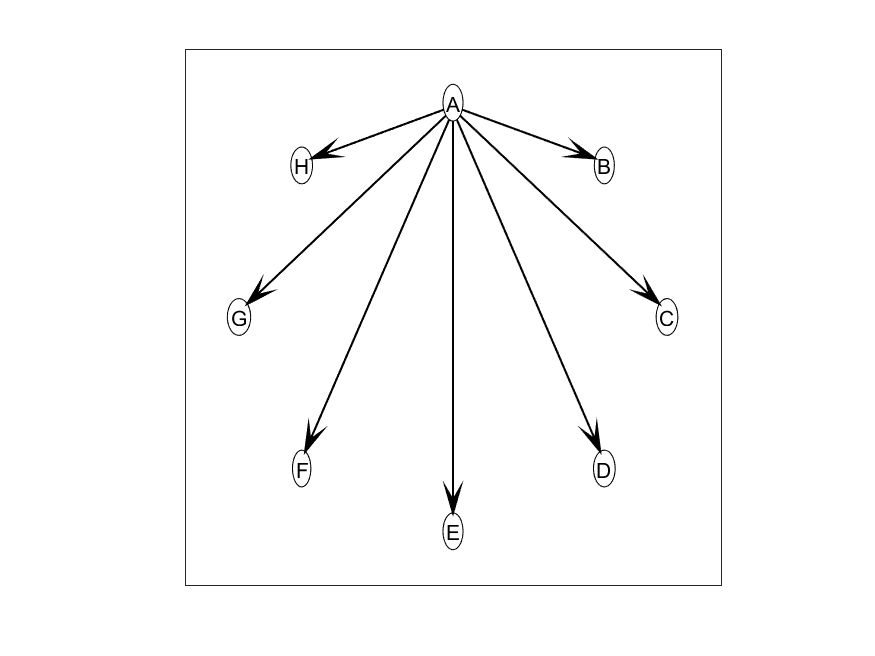}} 
\end{tabular}
\caption{Topology for the 8-node chain structure (left) and 8-node tree structure (right).}
\label{fig:chain-tree-structures}
\end{figure}

\begin{figure}
\centering
\includegraphics[scale = 0.6]{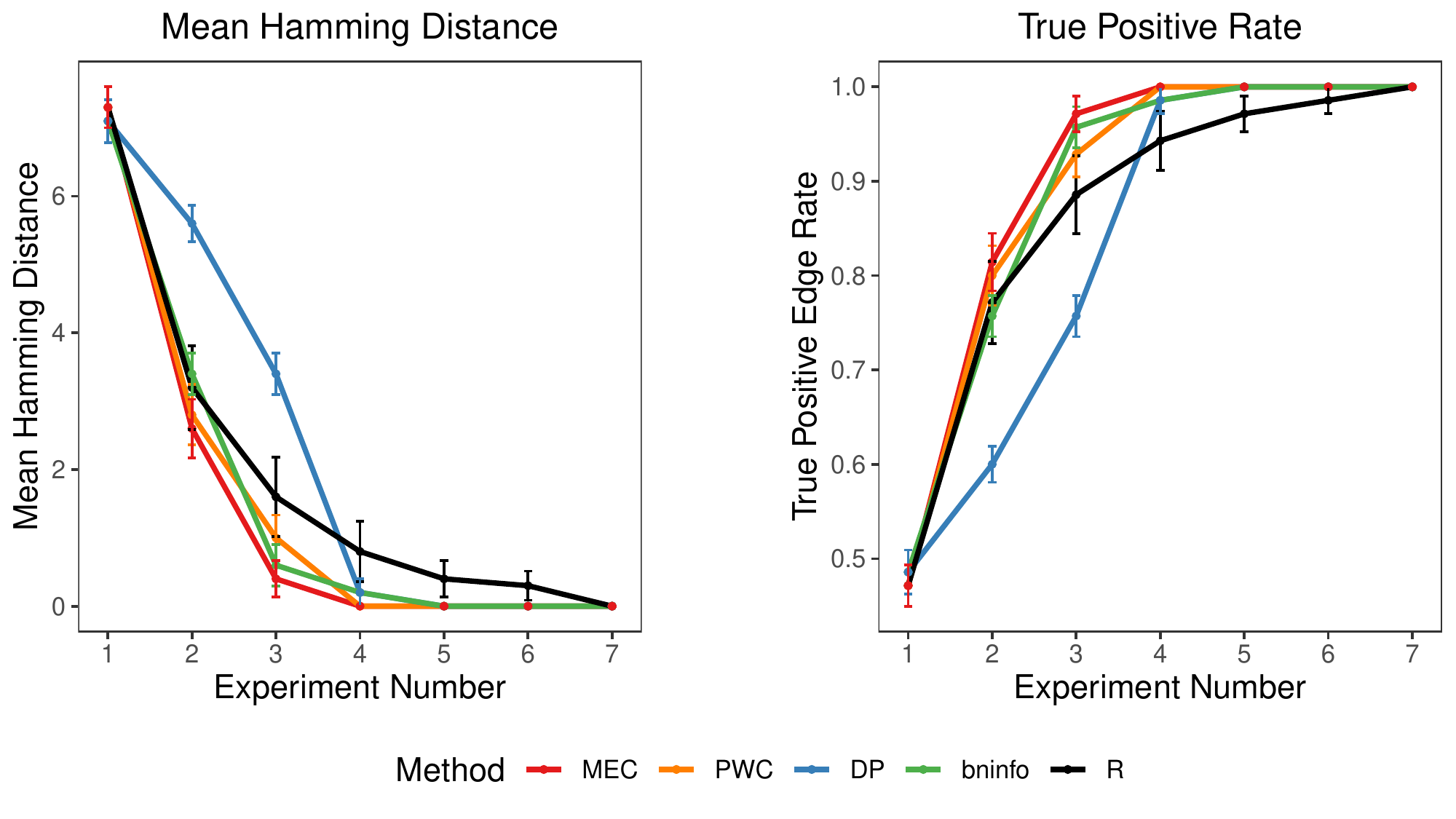}
\caption{Mean Hamming distance and mean TPR for the 8-node chain structure. MEC = our method with MEC partition scheme;  PWC = our method with pairwise child partition scheme; DP = dynamic programming method of \citet{LiLeong}; R = random learner.}
\label{fig:8-node-chain}
\end{figure}

\begin{figure}
\centering
\includegraphics[scale = 0.6]{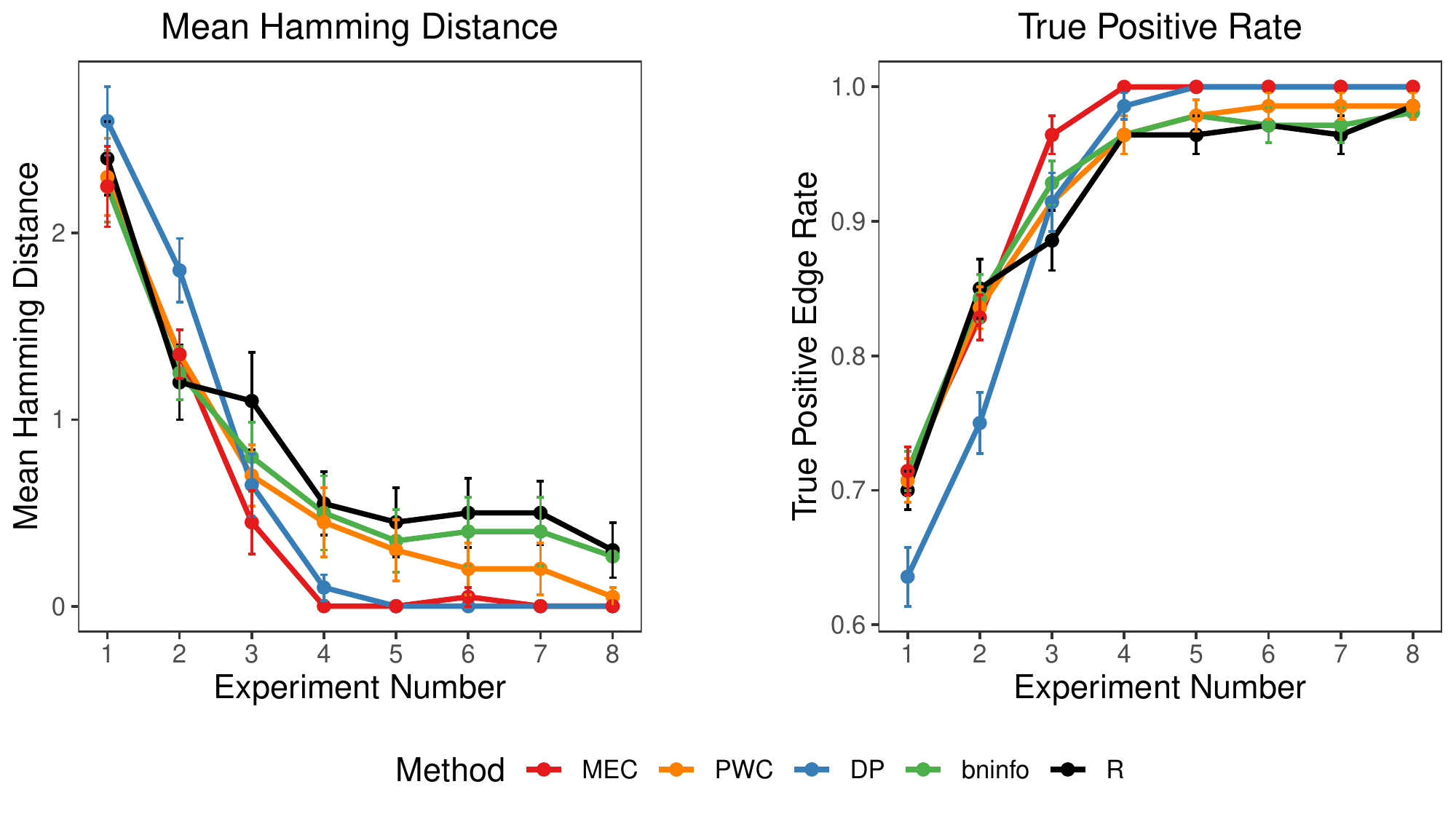}
\caption{Mean Hamming distance and mean TPR for the 8-node tree structure. MEC = our method with MEC partition scheme; PWC = our method with pairwise child partition scheme; DP = dynamic programming method of \citet{LiLeong}; R = random learner.}
\label{fig:8-node-tree}
\end{figure}

\section{Application to a cell-signaling network}\label{sec:application}
OED methods for graphical models are often developed with the goal of inferring biological networks, such as gene regulatory networks or cell-signaling networks \citep{Cho, Ness, Pournara, Sverchkov}. Especially in light of recent advances in the precision of gene-editing technologies, the ability to iterate between experimentation and analysis by adaptively selecting experiments is a promising avenue for reconstructing biological networks. In this section, we first apply our OED method, as well as the DP and bninfo methods, to human T-cell signaling data collected by \citet{Sachs05}. We then explore the performance of the methods on a simulated data set based on the Sachs network.  

\subsection{Analysis on real experimental data from the Sachs network}
The Sachs data set consists of concentration levels measured via flow cytometry for 11 proteins involved in activating the immune system. The true network describing the relationship between these proteins is unknown. Many network inference studies have explored this data set, but what they define as the benchmark graph often varies by 1-2 edges; this is likely because the biologists' consensus network is complex and contains a bidirectional relationship that would induce a cycle (Figure~\ref{fig:sachsNetwork}, left panel). Here, we use the benchmark network provided by \cite{Scutari} (Figure~\ref{fig:sachsNetwork}, right panel) as well as the discretized data set available via the \emph{bnlearn} package. A portion of the data, 1800 samples, was gathered under no targeted interventions and the remaining 3600 samples were collected after activating or inhibiting five signaling proteins: Mek12, Pip2, Akt, PKA, and PKC. (Mek12, Pip2, and Akt were each inhibited in 600 samples, PKA was activated in 600 samples, and PKC was inhibited in 600 samples and activated in an additional 600 samples.) 

We compared how well the MEC, DP, bninfo, and random intervention methods inferred the cell-signaling network over a series of six experiments. For all methods, the first experiment used the 1800 observational samples. For subsequent experiments, each method then chose from the set of five candidate interventions performed by \cite{Sachs05}.  Figure~\ref{fig:sachs} summarizes the results. While all methods end closer to the benchmark structure after accumulating data from the five interventions, no method performs particularly well. Curiously, the mean Hamming distance initially gets worse after the first couple of experiments before getting better.  See Figure~\ref{app:fig-adj-mat} (Appendix B) for a comparison of the benchmark adjacency matrix to the matrices estimated by the MEC, DP, and bninfo methods.

To determine the upper and lower bounds on performance for this data, we also considered all 120 possible permutations of the sequence of five manipulated nodes. Figure~\ref{fig:sachs} shows the results for the best- and worst-performing fixed sequences in dashed lines, as determined by mean Hamming distance averaged over the six experiments. By ``fixed" sequence, we mean the sequence of manipulated nodes was prespecified before the first experiment as opposed to adaptively or randomly chosen over the course of the experiments. Even the fixed sequence with the lowest mean Hamming distance over the six experiments (Figure~\ref{fig:sachs} orange dashed line) initially moves further from the benchmark network after the first intervention.    

There are several reasons why the methods perform differently on this data set than they did in simulation studies. First, if the true biological network consists of cycles, then the directed acyclic graphical models assumed by each of the algorithms would be misspecified. The consensus network determined by biologists suggests a short cycle among PIP3 $\rightarrow$ PLC$\gamma$ $\rightarrow$ PIP2, so model misspecification is a concern. \cite{Mooij2013} identified another reason why the model might be misspecified: \cite{Sachs05} used an experimental intervention that changed the \emph{activity} of the target proteins rather than directly intervening on the abundance of the protein. An intervention that affects the underlying topology of the network in ways other than only removing arrows into the manipulated protein differs from the type of edge-breaking intervention that our model assumes. Even if the acyclicity and edge-breaking intervention assumptions were not violated, the Hamming distance at the end of the six experiments was likely high because we were limited to interventions on five candidate proteins rather than all 11 proteins. Our results would also change if we used a different discretization of the data than the one provided by \cite{Sachs05}. However, given that \cite{Cho} used a continuous version of the Sachs data set for their Gaussian Bayesian network and encountered similar difficulties in network reconstruction, we do not believe our results would improve substantially if we used a different discretization of the data. 

We note that the performance of the bninfo method reported here differs from that published by \cite{Ness} for the Sachs network. This is likely due in large part to differences in the data sets used in our analyses. \cite{Ness} used 11,672 observational samples (they refer to this as ``historic data"), whereas we used only the 1800 observational samples provided in the \emph{bnlearn} R package since \cite{Ness} were unable to provide us with access to the larger data set they used in their analysis.     

\begin{figure}
	\centering
	\begin{minipage}{.5\textwidth}
		\centering
		\includegraphics[width=1.0\linewidth]{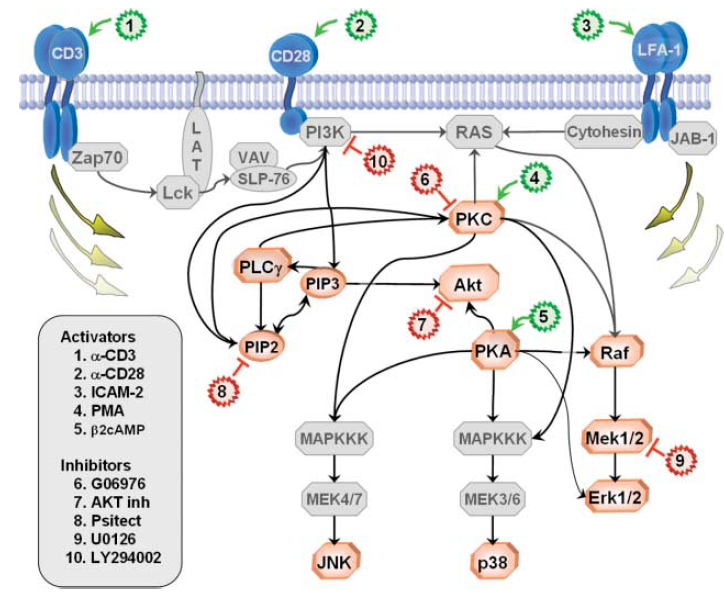}
	\end{minipage}%
	\begin{minipage}{.5\textwidth}
		\centering
		\includegraphics[width=1.0\linewidth]{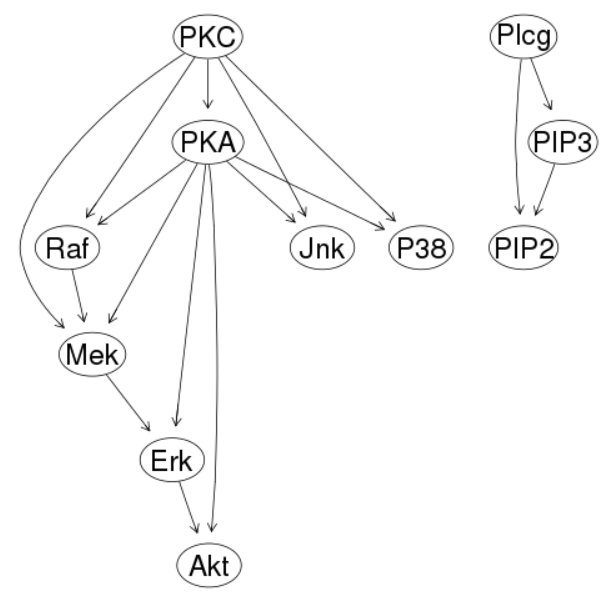}
	\end{minipage}
	\caption{Left: Signaling network diagram taken from \citet{Sachs05}. Reprinted with permission from AAAS. Right: Network structure from the Bayesian Network Repository available in the \emph{bnlearn} R package (\cite{Scutari}) used here as the benchmark network.}
	\label{fig:sachsNetwork}
\end{figure} 

\begin{figure}
	\centering
		\includegraphics[width=1.0\linewidth]{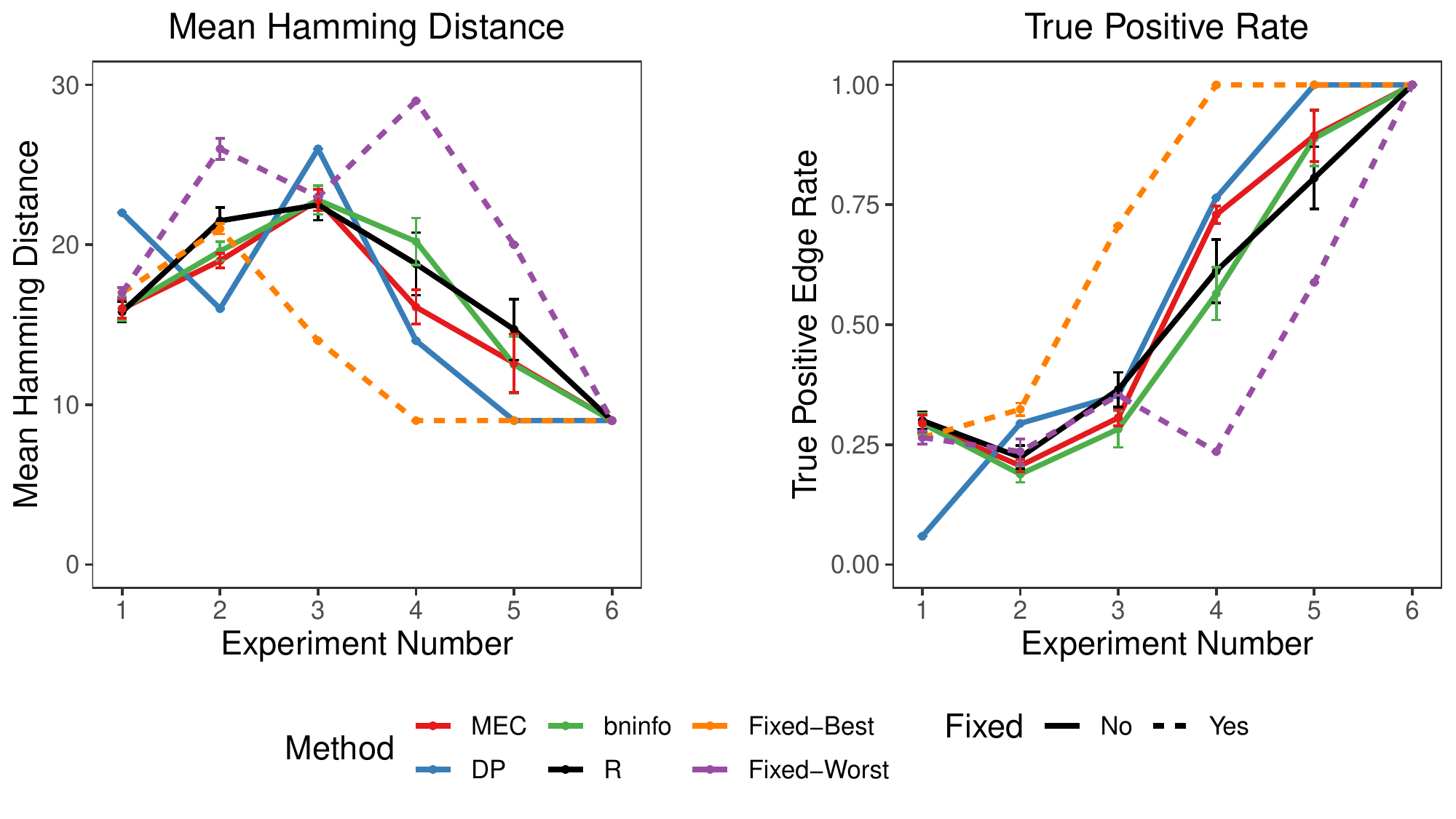}
	\caption{Mean Hamming distance and true positive rate on the cell- signaling data from the 11-node Sachs network data. MEC = our method with MEC partition scheme; DP = dynamic programming method of \citet{LiLeong}; R = random learner. Dashed lines represent the best-case and worst-case fixed sequence of interventions. }
	\label{fig:sachs}
\end{figure}

\subsection{Analysis on simulated data from the Sachs benchmark network}
To understand whether the poor performance seen on the Sachs data was due to misspecification, we tried simulating data from the benchmark network to see how the methods perform when the model assumptions hold. Figure \ref{fig:simSachs} shows the results of a simulation study comparing the same methods as in Figure \ref{fig:sachs}, but using simulated data generated from the benchmark network in Figure \ref{fig:sachs}, using a conditional probability table estimated from the Sachs data and available in the \emph{bnlearn} R package. The MEC method performed well and fell within a Hamming distance of one from the benchmark network by the fourth experiment, on average. Bninfo also initially performed well, but then plateaued sooner than the other OED methods, failing to reach a Hamming distance of zero or TPR of 1 by the seventh experiment. The higher mean Hamming distance of the DP method for the first two experiments arose from a combination of both a lower true positive rate and higher false negative rate than the other methods. This is likely because the DP prior over graphs tends to encourage sparsity, but the ground truth network contains nodes like PKC and PKA with five and six children, respectively.

\begin{figure}
\centering
\includegraphics[width=1.0\linewidth]{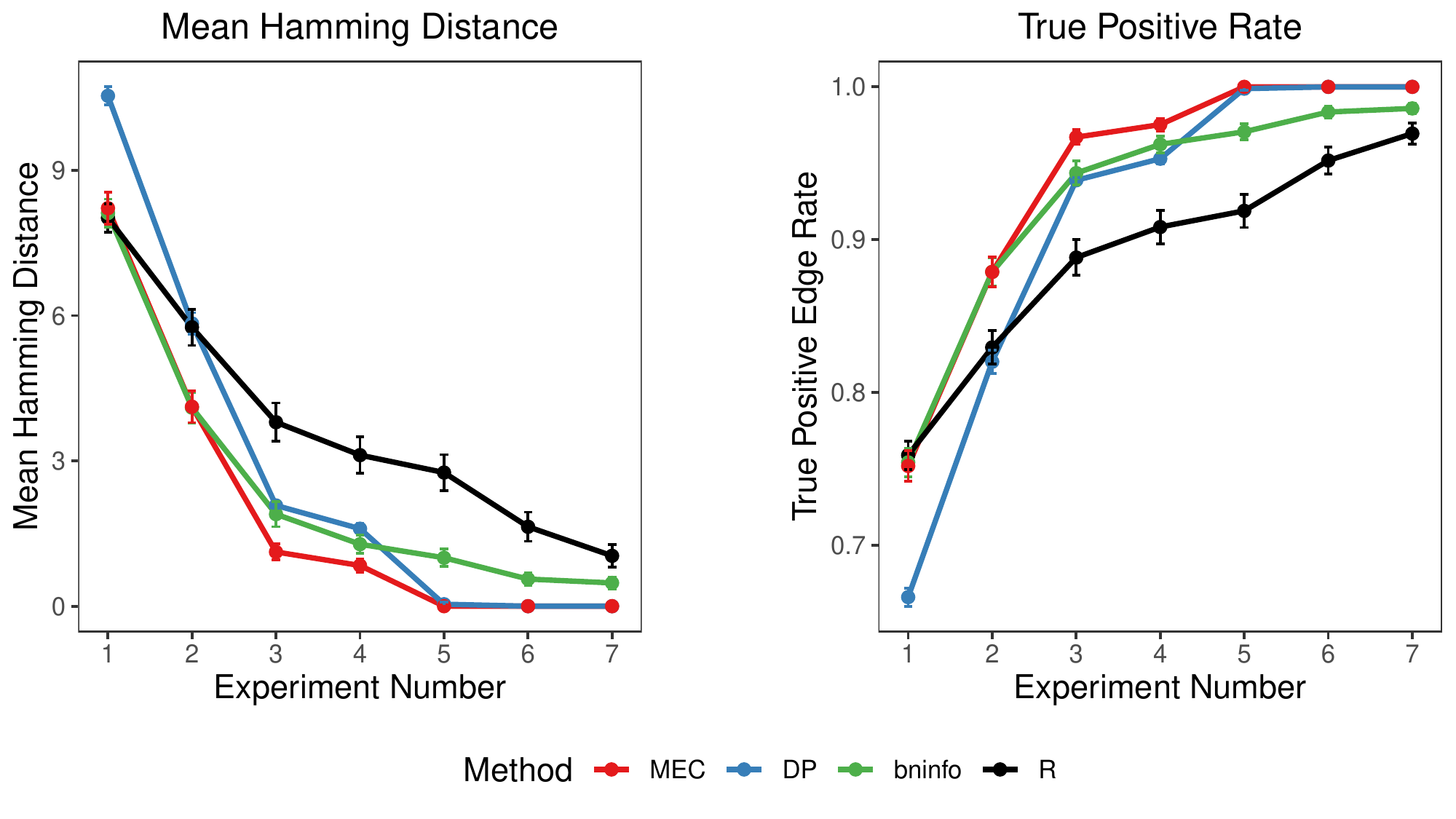}
\caption{Mean Hamming distance and true positive rate on the cell-signaling data from the simulated 11-node Sachs network data.
MEC = our method with MEC partition scheme; DP = dynamic programming method of \citet{LiLeong}; R = random learner.}
\label{fig:simSachs}
\end{figure}

\section{Conclusion}\label{sec:conclusion}
We presented a novel Bayesian OED methodology for optimizing the experiment selection process in a computationally tractable way. The core of the method is a criterion for selecting the experiment that is expected to yield the greatest reduction in posterior entropy. We found that the method efficiently infers causal relationships in networks with various topologies, with the greatest gains in information coming from the first few optimally-chosen interventions. We provided a theoretical justification for using Markov equivalence classes as the choice of partition in our method, and in simulations, we found that this entropy criterion generally performs well empirically. 

Currently, our method is limited to networks with less than 25 nodes due to the super-exponential growth in the number of candidate graphs with respect to the number of nodes and the computational limits of the DP-based MCMC proposals. 
% Thus, we expect this framework will be most useful for researchers interested in determining the causal relationships among a relatively focused group of variables, such as an experimental biologist studying the genes in a specific pathway or biological process.
Scaling up the method to work on larger networks is an area for future work.

The difficulty that many OED and active learning methods, including our own, have in inferring the Sachs network suggests that additional research is needed on ways of relaxing the acyclicity assumption and being more robust to model misspecification in general. Additionally, the OED and active learning fields would benefit from additional data sets similar in nature to the Sachs data with a mix of observational and intervention data. These will be helpful for evaluating and comparing OED methods. As recent advances in gene-editing technologies make targeted interventions more feasible, we expect these types of data sets will become more widely available, and the demand for OED methods in the biological sciences will grow in tandem.

%%%%%%%%%%%%%%%%%%%%%%%%%%%%%%%%%%%%%%%%%%%%%%
%% Supplementary Material, if any, should   %%
%% be provided in {supplement} environment  %%
%% with title and short description.        %%
%%%%%%%%%%%%%%%%%%%%%%%%%%%%%%%%%%%%%%%%%%%%%%
%  \begin{supplement}
%  \stitle{Appendix A: Theoretical results}
%  \sdescription{Proofs of Theorem 4.2 and Theorem 4.3.}
%  \end{supplement}
 
%  \begin{supplement}
%  \stitle{Appendix B: Additional results for the Sachs network}
%  \sdescription{Adjacency matrices across methods for the Sachs network.}
%  \end{supplement}

\begin{appendix}

\section{Theory}

\begin{lemma}
\label{lemma:doob-entropy}
Suppose $(\theta,\nu)\sim\pi$, $X_1,\ldots,X_N|\theta,\nu\sim P_{\theta,\nu}$ i.i.d., and $f(\theta)$ satisfies Conditions~\ref{condition:identifiable} and \ref{condition:finite}.  Then $H(f(\theta) \mid X_{1:N}) \to 0$ as $N\to\infty$.
\end{lemma}
\begin{proof}
Since $g(P_{\theta,\nu}) = f(\theta)$ a.s.\ under the prior, then the same also holds a.s.\ under the posterior. Thus, for any value $y$ in the range of $f$,
\begin{align*}
p(f(\theta)=y \mid X_{1:N}) &= 
p(g(P_{\theta,\nu})=y \mid X_{1:N}) = 
\mathrm{E}(\mathds{1}(g(P_{\theta,\nu})=y) \mid X_{1:N}) \\
&\xrightarrow[N\to\infty]{\mathrm{a.s.}} \mathds{1}(g(P_{\theta,\nu})=y)
\overset{\mathrm{a.s.}}{=} \mathds{1}(f(\theta)=y)
\end{align*}
where $\mathds{1}(\cdot)$ is the indicator function.
Here, the limiting value is a random variable in which $(\theta,\nu)\sim\pi$, whereas $(\theta,\nu)$ is integrated out in the probabilities/expectations.
Thus, since the range of $f$ is finite,
$$ -\sum_y p(f(\theta)=y\mid X_{1:N})\log p(f(\theta)=y\mid X_{1:N})
\xrightarrow[N\to\infty]{\mathrm{a.s.}} 0,$$
with the convention that $0 \log 0 = 0$.
Since the entropy of a random variable on a finite set is bounded, then
by the dominated convergence theorem,
$$ H(f(\theta)\mid X_{1:N}) = \mathrm{E}\Big( -\sum_y p(f(\theta)=y\mid X_{1:N})\log p(f(\theta)=y\mid X_{1:N})\Big)
\xrightarrow[N\to\infty]{} 0. $$
This completes the proof.
\end{proof}

\begin{lemma}
\label{lemma:likelihood-equivalence}
Let $X$ and $\theta$ be random variables with joint density $p(x,\theta)$.
Suppose $f(\theta)$ is a discrete random variable such that:
for any $\theta,\theta'$, if $f(\theta)=f(\theta')$ then for all $x$, $p(x|\theta) = p(x|\theta')$.
Then $X \indep \theta \mid f(\theta)$.
\end{lemma}
\begin{proof}
Let $Y = f(\theta)$.  Let $y$ be any value such that $p(y) > 0$, and define $A = \{\theta : f(\theta) = y\}$.
Then for all $\theta,\theta'\in A$, we have $p(x|\theta,y) = p(x|\theta) = p(x|\theta')$ by assumption. Hence,
\begin{align*}
    p(x |y) &= \int p(x |\theta,y)\, p(\theta|y) \lambda(d\theta) 
    = \int_A p(x |\theta,y)\, p(\theta|y) \lambda(d\theta) \\
    &= \int_A p(x |\theta')\, p(\theta|y) \lambda(d\theta) 
    = p(x|\theta') = p(x\mid\theta,y)
\end{align*}
where $\theta,\theta'\in A$, and $\lambda(d\theta)$ is the dominating measure for $p(\theta|y)$.
Thus, $p(x|y)p(\theta|y) = p(x|\theta,y)p(\theta|y) = p(x,\theta|y)$.
\end{proof}

\begin{proof}[Proof of Theorem~\ref{theorem:cond-indep}]
For notational brevity, denote $e=i'$ and $N=N'$, and define $f(G) = (f_e(G),f_{i_1}(G),\ldots,f_{i_K}(G))$.
Suppose we can show that for all $G_1$ and $G_2$, 
if $f_e(G_1) = f_e(G_2)$ then $X_{1:N}^{e} | G_1$ is equal in distribution to $X_{1:N}^{e} | G_2$.
Then the result will follow by Lemma~\ref{lemma:likelihood-equivalence}, since if $f(G_1) = f(G_2)$,
then in particular, $f_e(G_1)=f_e(G_2)$.

We show that if $f_e(G_1) = f_e(G_2)$ then $X_{1:N}^{e} | G_1 \overset{\mathrm{d}}{=} X_{1:N}^{e} | G_2$.
First observe that the assumed prior factors as $\pi(\beta|G) = \prod_{i=1}^V p(\beta_i|G) p(\beta_i^*)$, 
and therefore, since node $e$ has no parents in $G^{e}$,
\begin{align}\label{eqn:marginal-likelihood-equivalence-proof}
    p(X_{1:N}^{e} = x_{1:N} \mid G) = p(X_{1:N} = x_{1:N} \mid G^{e}) \frac{p(X_{e,1:N}^{e} = x_{e,1:N} \mid G)}{p(X_{e,1:N} = x_{e,1:N} \mid G^{e})}
\end{align}
where $x_{e,1:N}$ denotes $(x_{e,n} : n = 1,\ldots,N)$.
Since $e$ has no parents in $G^e$, $p(X_{e,1:N} = x_{e,1:N} \mid G^{e})$ does not depend on $G$.
Similarly, since $p(X_{e,1:N}^{e} = x_{e,1:N} \mid G) = \int \big(\prod_{n=1}^N p^*(x_{e,n} \mid \beta_e^*)\big) p(\beta_e^*) d\beta_e^*$,
this does not depend on $G$ either.

By Theorem 5 of \citet{Heckerman1995}, the BDeu metric is likelihood equivalent, which implies that
for any $G_1,G_2$ such that $f_e(G_1) = f_e(G_2)$, we have $p(X_{1:N} = x_{1:N} \mid G_1^{e}) = p(X_{1:N} = x_{1:N} \mid G_2^{e})$.
Therefore, applying these invariance properties to Equation~\ref{eqn:marginal-likelihood-equivalence-proof}, we see that 
if $f_e(G_1) = f_e(G_2)$ then $p(X_{1:N}^{e} = x_{1:N} \mid G_1) = p(X_{1:N}^{e} = x_{1:N} \mid G_2)$.
This completes the proof.
%for all $x_{1:N'}$ we have $p(X_{1:N'}^{e} = x_{1:N} \mid G) = p(X_{1:N'}^{e} = x_{1:N} \mid G) = $
\end{proof}

\begin{proof}[Proof of Theorem~\ref{theorem:identifiable}]
For notational brevity, denote $e=i'$ and $N=N'$.
A distribution $P$ is said to be \textit{faithful} to a graph $G$ if the set of conditional independence relations that are true for $P$ are all and only those implied by $G$.  More precisely, given a distribution $P$ on $(X_1,\ldots,X_V)$, define $g(P) = (\mathds{1}(X_A\indep_P X_B\mid X_C) : A,B,C\subseteq \{1,\ldots,V\}\}$, that is, $g(P)$ is a binary vector indicating which conditional independence properties hold under $P$.  Meanwhile, given a DAG $G$ on $\{1,\ldots,V\}$, define $f(G) = (\mathds{1}(X_A\indep_G X_B\mid X_C) : A,B,C\subseteq \{1,\ldots,V\}\}$, that is, $f(G)$ is a binary vector indicating which conditional independence properties are implied by $G$ according to the d-separation criterion.  Then $P$ is faithful to $G$ if and only if $g(P) = f(G)$.

Let $B(G)$ be the support of the prior $\pi(\beta|G)$. Let $\lambda_G$ denote the dominating measure of $\pi(\beta|G)$ on $B(G)$.
(Colloquially, one might refer to $\lambda_G$ as ``Lebesgue measure on $B(G)$'', but technically there are sum-to-one constraints on the probability vectors, so technically it is Lebesgue measure on a lower-dimensional subspace.)
By Theorem 7 of \citet{Meek1995}, for any $G$, the set $\{\beta\in B(G) : P_{G,\beta} \text{ is not faithful to } G\}$ has measure zero under $\lambda_G$.  In particular,  $\{\beta^e\in B(G^e) : P_{G^e,\beta^e} \text{ is not faithful to } G^e\}$ has measure zero under $\lambda_{G^e}$.
Let $\pi^e$ denote the distribution of $(G^e,\beta^e)$ when $(G,\beta)\sim p(G,\beta \mid D)$.
Suppose we can show that $\pi^e(\beta^e | G^e)$ has a density with respect to $\lambda_{G^e}$.
Then it follows that, almost surely under $\pi^e$, $P_{G^e,\beta^e}$ is faithful to $G^e$.
In other words, $g(P_{G^e,\beta^e}) = f(G^e)$ almost surely when $(G,\beta)\sim p(G,\beta\mid D)$.
The conclusion of the theorem follows since, by construction, there is a one-to-one mapping between $f(G^e)$ and $f_e(G)$.

To complete the proof, we need to show that $\pi^e(\beta^e | G^e)$ has a density with respect to $\lambda_{G^e}$,
or in mathematical notation, $\pi^e(\beta^e | G^e) \ll \lambda_{G^e}$.
To see this, first observe that $\pi(\beta|G) \ll \lambda_G$, and thus, $p(\beta|G,D) \ll \lambda_G$.

Next, we argue that $p(\beta^e|G,D) \ll \lambda_{G^e}$. 
Recall that $\beta^e$ is a function of $\beta$ that is obtained by copying $\beta$ and then (i) putting $\beta_e^*$ in place of $\beta_e$, and (ii) putting $\beta_{e 1}$ in place of $\beta_e^*$.
Let $A^e \subseteq B(G^e)$ such that $\lambda_{G^e}(A^e) = 0$, 
and define $A = \{\beta\in B(G) : \beta^e \in A^e\}$.  Then $\lambda_G(A) = 0$, since $\lambda_G$ is the product of identical measures (colloquially, ``Lebesgue measure on the probability simplex'') for each $\beta_{i j}$ and each $\beta_i^*$.
Hence, $p(\beta^e\in A^e \mid G,D) = p(\beta\in A \mid G,D) = 0$.  This implies that $p(\beta^e|G,D) \ll \lambda_{G^e}$. 

Letting $H$ be a function of $G$ defined by $H = G^e$, we have
\begin{align*}
    p(\beta^e|G^e,D) = p(\beta^e|H,D) = \sum_G p(\beta^e|G,H,D)p(G|H,D) = \sum_{G\,:\,G^e=H} p(\beta^e|G,D)p(G|H,D),
\end{align*}
and thus, $p(\beta^e|G^e,D) \ll \lambda_{G^e}$.
Since $\pi^e(\beta^e|G^e)$ is just another way of writing $p(\beta^e|G^e,D)$, then $\pi^e(\beta^e|G^e) \ll \lambda_{G^e}$, as claimed.
\end{proof}

\section{Adjacency matrices across methods for Sachs network}\label{appendixC-adj-mat}

\begin{figure}
    \centering
    \begin{tabular}{cc}
     \includegraphics[scale = 0.55]{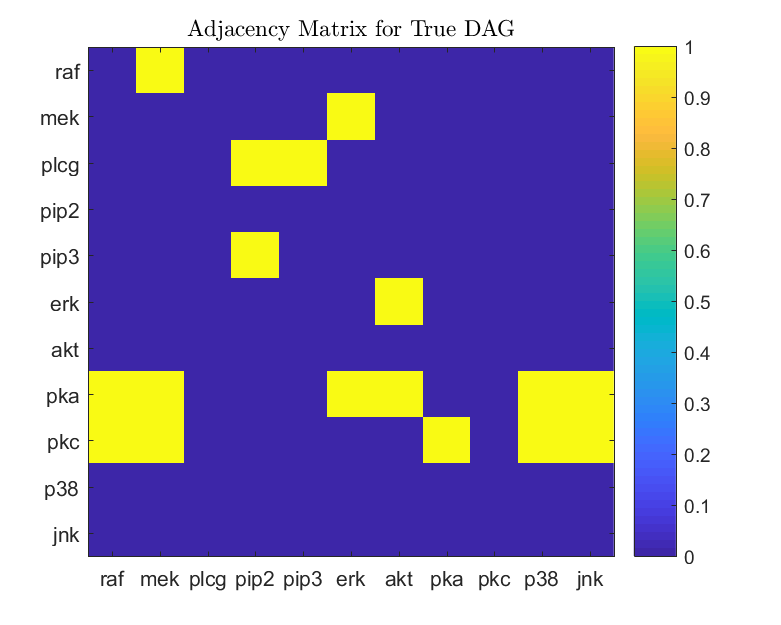}    & 
     \includegraphics[scale = 0.55]{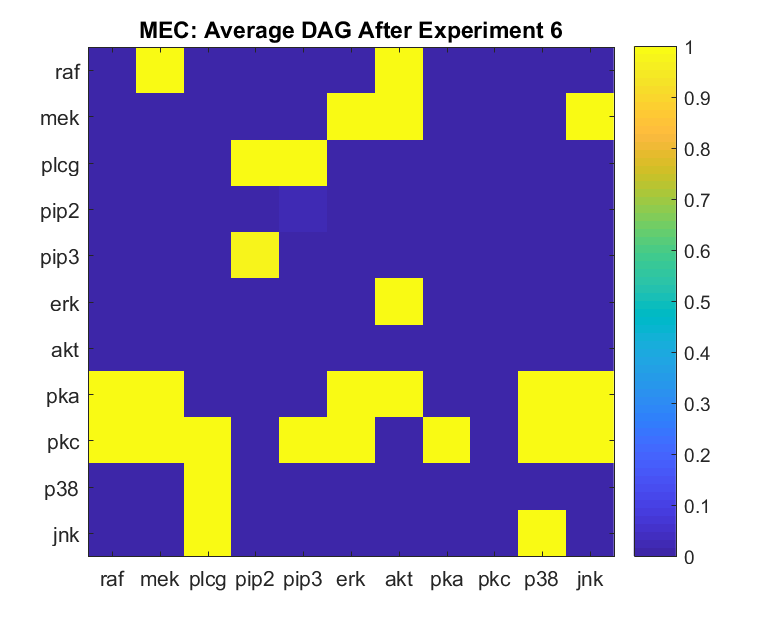} \\         \includegraphics[scale = 0.55]{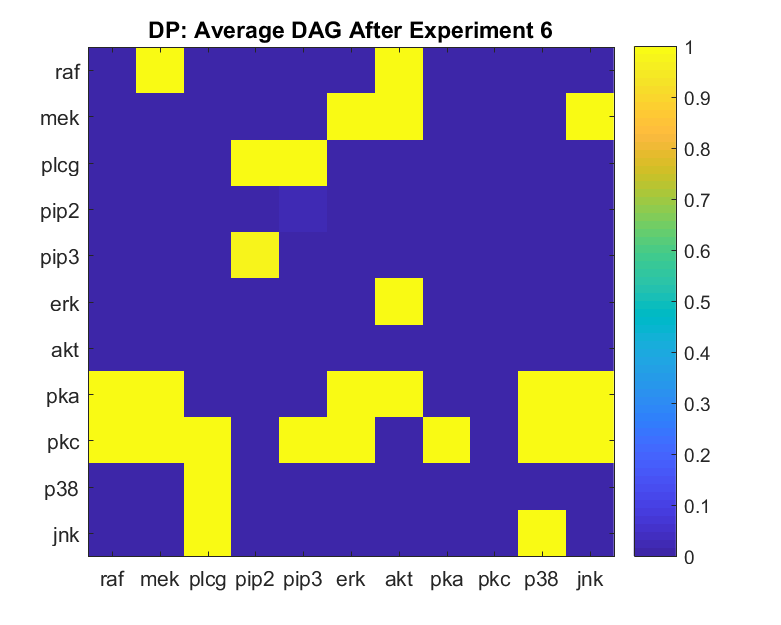}    & 
     \includegraphics[scale = 0.55]{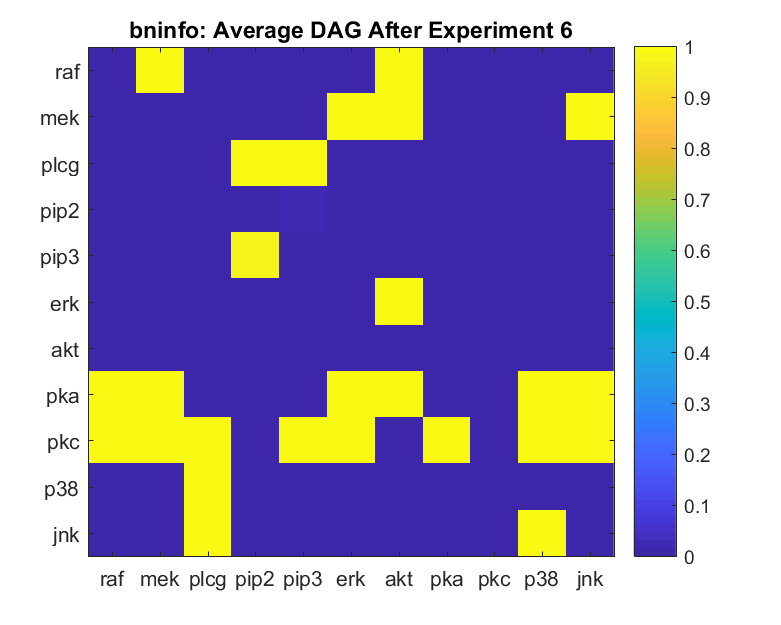} \\
    \end{tabular}
    \caption[Adjacency matrices after the sixth experiment for Sachs data]{Adjacency matrices after the sixth experiment for the Markov equivalence class (MEC), dynamic programming (DP), and bninfo methods on the Sachs data. Top left: benchmark DAG provided by \cite{Scutari}. Rows denote parent nodes and columns denote child nodes. Yellow indicates presence of a directed edge while blue indicates absence of an edge.}
    \label{app:fig-adj-mat}
\end{figure}
 The adjacency matrices in Figure~\ref{app:fig-adj-mat} for the MEC, DP, and bninfo methods all converge to the same DAG after experiment six. Note, however, that this estimated DAG differs from the structure of the benchmark DAG shown in the top left panel of Figure~\ref{app:fig-adj-mat}. The DAGs differ by a Hamming distance of nine, which upon further inspection is due to the estimated DAGs including the following false positive edges: 
    \begin{enumerate}
        \item raf $\rightarrow$ akt 
        \item mek $\rightarrow$ akt
        \item mek $\rightarrow$ jnk
        \item pkc $\rightarrow$ plcg
        \item pkc $\rightarrow$ pip3
        \item pkc $\rightarrow$ erk 
        \item p38 $\rightarrow$ plcg
        \item jnk $\rightarrow$ plcg
        \item jnk $\rightarrow$ p38
    \end{enumerate}
\end{appendix}

%% ** The bibliography **
\newpage
\bibliographystyle{unsrt}
\bibliography{refs}% place <bib-data-file> in ./bib folder 

\end{document}